\documentclass[11pt]{article}

\usepackage[left=1.5cm, right=2cm]{geometry}
\usepackage{graphicx}  
\usepackage{bm}        
\usepackage{epstopdf}
\usepackage{placeins}
\usepackage{calc}
\usepackage{graphicx}  
\usepackage{dcolumn}   
\usepackage{amssymb}   
\usepackage{amsmath}
\usepackage{comment}
\usepackage{color}
\usepackage{hyperref}
\usepackage{tabularx}
\usepackage{algorithm}
\usepackage{algorithmicx}
\usepackage{algpseudocode}
\usepackage{subfig}
\usepackage{multirow}
\usepackage{tikz,ifthen}
\usepackage{enumitem}
\usepackage{cite}
\usepackage{booktabs}
\usepackage{colortbl}
\usepackage{cleveref}

\usepackage{amsmath,amsthm}

  \usepackage{pgfplots}
  \pgfplotsset{compat=newest}
  \usetikzlibrary{plotmarks}
  \usetikzlibrary{arrows.meta}
  \usepgfplotslibrary{patchplots}
  \usepackage{grffile}

\newcommand{\Uone}{{\rm U}(1)}

\newtheorem{myremark}{Remark}

\newtheorem{myprop}{Proposition}
\newtheorem{mylemma}{Lemma}
\newtheorem{definition}{Definition}
\newtheorem{mycorol}{Corollary}

\DeclareMathOperator{\Tr}{Tr}

\DeclareMathOperator{\rank}{rank}
\DeclareMathOperator{\ddiag}{ddiag}

\DeclareMathOperator{\diag}{diag}

\DeclareMathOperator*{\minimize}{minimize\,}
\DeclareMathOperator*{\maximize}{maximize\,}

\DeclareMathOperator*{\argmax}{argmax\,}

\newcommand{\rmd}{\mathrm{d}}

\newcommand{\rmi}{\mathrm{i}}
\newcommand{\degree}{\mathrm{d}}

\pgfplotsset{compat=1.3, every axis label/.append style={font=\scriptsize}, tick label style={font=\tiny}, legend style={font=\scriptsize}}

\definecolor{lavander}{cmyk}{0,0.48,0,0}
\definecolor{violet}{cmyk}{0.79,0.88,0,0}
\definecolor{burntorange}{cmyk}{0,0.52,1,0}

\def\oran{orange!30}
\tikzstyle{Node2}=[draw,circle,burntorange, left color=\oran,
                       text=violet,minimum width=18pt]
\tikzstyle{Node1}=[draw,circle,burntorange, left color=\oran,
                       text=violet,minimum width=6pt]
\tikzstyle{light}=[draw,circle,violet, left color=violet,
                      text=violet,minimum width=10pt]

\tikzset{
  LabelStyle/.style = { rectangle, rounded corners, draw,
                        minimum width = 2em, fill = yellow!50,
                        text = red, font = \bfseries },
  VertexStyle/.append style = { inner sep=5pt,
                                font = \Large\bfseries},
  EdgeStyle/.append style = {->, bend left} }

\newlength\figureheight 
\newlength\figurewidth

\title{Ellipsoidal embeddings of graphs}

\author{ Micha\"el Fanuel\thanks{Université de Lille, CNRS, Centrale Lille, UMR 9189-CRIStAL, F-59000 Lille, France, \texttt{{michael.fanuel@univ-lille.fr}}} \and Antoine Aspeel\thanks{Electrical Engineering and Computer Science Department, Univ.\ of Michigan, Ann Arbor, USA, \texttt{{antoinas@umich.edu}}}
\and Michael T. Schaub\thanks{RWTH Aachen University, North-Rhine Westphalia, Germany, \texttt{{schaub@cs.rwth-aachen.de}}}  
\and Jean-Charles Delvenne\thanks{ICTEAM and CORE, Universit\'e catholique de Louvain, Belgium. \texttt{{jean-charles.delvenne@uclouvain.be}}}
}

\date{\today}

\begin{document}
\maketitle

\begin{abstract}
    Due to their flexibility to represent almost any kind of relational data, graph-based models have enjoyed a tremendous success over the past decades.
    While graphs are inherently only combinatorial objects, however, many prominent analysis tools are based on the algebraic representation of graphs via matrices such as the graph Laplacian, or on associated graph embeddings.
    Such embeddings associate to each node a set of coordinates in a vector space, a representation which can then be employed for learning tasks such as the classification or alignment of the nodes of the graph.
    As the geometric picture provided by embedding methods enables the use of a multitude of methods developed for vector space data, embeddings have thus gained interest both from a theoretical as well as a practical perspective.
    Inspired by trace-optimization problems, often encountered in the analysis of graph-based data, here we present a method to derive ellipsoidal embeddings of the nodes of a graph, in which each node is assigned a set of coordinates on the surface of a hyperellipsoid.
    Our method may be seen as an alternative to popular spectral embedding techniques, to which it shares certain similarities we discuss.
    To illustrate the utility of the embedding we conduct a case study in which we analyse synthetic and real world networks with modular structure, and compare the results obtained with known methods in the literature.
\end{abstract}

\section{Introduction: Graphs and Embeddings}
Graphs enable us to conceptualise many different complex systems in a simple and compact manner.
A graph $\mathcal G(\mathcal V, \mathcal E)$ consist of a set of vertices (or nodes) $\mathcal V$ and a set of edges (or links) $\mathcal E$.
The node set $\mathcal V$ is used to denote the entities present in the system, and the edges in the set $\mathcal E$ designate the interactions between these entities.
By specifying a suitable set of nodes and edges, most types of relational data can be abstracted as a graph.
Accordingly, graphs have enjoyed an enormous success as mathematical modelling tools over the last decades, pervading essentially all areas of science~\cite{Strogatz2001,Newman2003,Boccaletti2006,Arenas2008a,Dorogovtsev2008}, from neurobiology~\cite{Sporns2009} to statistical physics~\cite{AlbertBarabasi}.

Arguably, a large part of the success of graphs as modelling tools is due to the minimal, yet versatile mathematical structure of graphs. 
We may enrich simple graphs for additional modelling flexibility, e.g., by allowing for weightings or directionality of the edges, or adding some form of multilayer structure.
Yet, when discussing graphs we often do not think of them in a purely combinatorial fashion. 
Rather, we tend to reason about graphs in terms of their algebraic representations as matrices, such as an adjacency matrix or a Laplacian; or we consider them in the form of visualizations via diagrams.

Though not inherent to the definition of graphs, in practice, both the algebraic and the visual representation are undeniably important for theory and applications.
Algebraic representations of graphs are, for instance, essential for computations and provide links to tools from matrix theory such as spectral analysis that enables a richer understanding of graphs.
Similarly, when talking about specific graph structures such as clusters, we often provide geometrical pictures that are supposed to convey the graph structure visually.
However, finding a good visualization of a graph in such a Euclidean space is not an easy task, as the (in)famous `hairball' pictures encountered when visualizing many large graphs highlight.

Before defining formally ellipsoidal embeddings, we provide below a first intuitive picture from the perspective of graph-drawing.

\subsection{The ellipsoidal embedding}\label{sec:additional_interpretations}
We identify the node-set $\mathcal{V}$ of a graph with the natural numbers $\{1,\dots,n\}$. 
We want to represent the vertex $i$ with a \emph{row} vector $h_i \in \mathbb{S}^{d_0-1}\subset \mathbb{R}^{d_0}$ for all $i=1,\dots, n$, where $d_0\geq 2$ is fixed. The integer $d_0$ is \emph{a priori} the dimension of the embedding space, as $\mathbb{S}^{d_0-1}$ denotes the $d_0-1$ dimensional Euclidean unit sphere.
The vectors $h_i$ are obtained by minimizing the energy 
\begin{equation}
E(H) = -\sum_{i,j=1}^n M_{ij} \langle h_i, h_j\rangle,\label{eq:EnergyMin}
\end{equation}
which depends on their Euclidean dot product $\langle h_i, h_j\rangle = h_i h_j^\top$ and where $M$ is a descriptor matrix of the graph, such as a modularity matrix or a Laplacian-based matrix as we detail in what follows.

\begin{figure}[!htb]
\centering
\begin{minipage}{1.\textwidth}
\centering
\includegraphics[trim={5.5cm 8.5cm 5.5cm 8.5cm},clip,scale = 0.5]{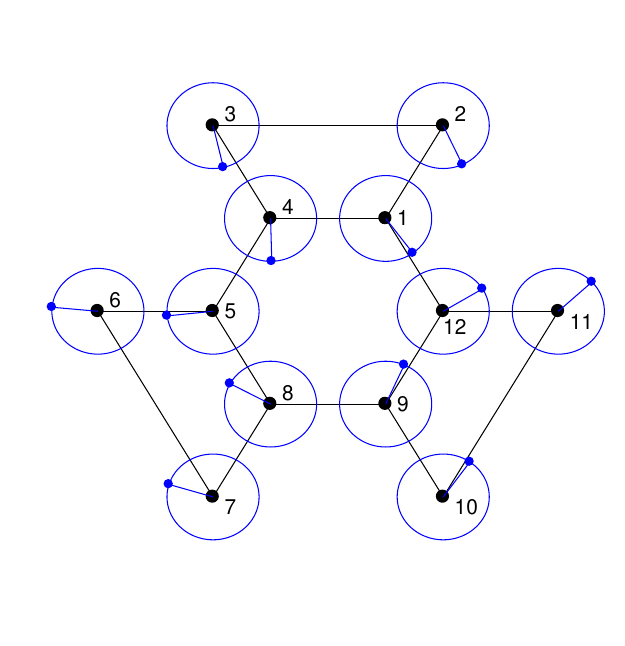}
\includegraphics[trim={4.5cm 8.5cm 4.5cm 8.5cm},clip,scale = 0.5]{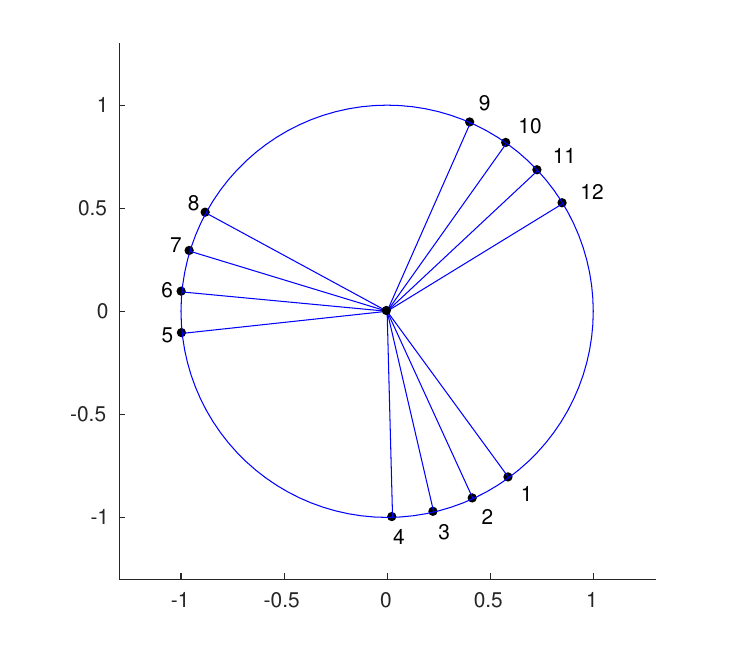}
\end{minipage}
\caption{\textbf{Modularity-based spherical embedding.}
Left: a spherical embedding -- see \cref{Def:Embedding} -- on the depicted graph (black lines), which yields a coordinate vector $h_i$ for each node $i$,  interpreted as a `spin'  attached to each node (blue). Right: each of these spins $h_i$ is drawn on the same hypersphere, giving rise to the depicted spherical embedding. The alignment of the spins $h_i$ reflects some of the neighbourhood structure in the graph. 
Here we chose $d_0 = 2$ and the graph descriptor matrix is the modularity $M=Q$; see \cref{e:Q_matrix}.}
\label{Fig:Foam}
\end{figure} 

Intuitively, the energy~\eqref{eq:EnergyMin} is minimized if nodes in the graph for which $M_{ij}>0$ are positioned close to each other on the sphere $\mathbb{S}^{d_0-1}$, while pairs of nodes for which $M_{ij}<0$ will repel each other. Unlike the usual force directed visualizations which place the nodes of a graph in the plane with attracting and repelling forces, the positions of the nodes are here constrained to be within a compact set and the force between two nodes $i$ and $j$ depends on the angle between the position vectors $h_i$ and $h_j$. The nature of the coupling between the nodes is dictated by the choice of the matrix $M$, which is taken to be a descriptor matrix of the graph, such as the modularity matrix $Q$ or a normalized Laplacian $\mathcal{L}$ that are discussed hereafter. Typically, $M_{ij}>0$ (resp. $<0$) if $i$ and $j$ are strongly (resp. loosely) connected.

To illustrate the embedding, we consider the example graph shown in Figure~\ref{Fig:Foam}: a toy network arranged in a set of 3 groups of 4 nodes (left).
By computing a spherical embedding of this graph, we obtain a coordinate vector for each node, interpreted as a `spin' variable valued on $\mathbb{S}^1$; see Figure~\ref{Fig:Foam} left for which $d_0 = 2$. Alternatively, we can consider all those coordinates on a single hypersphere as in Figure~\ref{Fig:Foam} right, illustrating the here proposed embedding. As should be apparent from Figure~\ref{Fig:Foam}, neighbouring nodes that are more tightly coupled in the graph tend to align their spins.
For convenience, the vectors $h_i$ are viewed as the rows of a matrix $H\in \mathbb{R}^{n\times d_0}$, that is $h_i = H_{i\ast}$. Then, the energy minimization~\eqref{eq:EnergyMin} can be rephrased as the following maximization problem
\begin{equation}
\maximize_{H\in \mathbb{R}^{n\times d_0}}\Tr\Big(H^\top M H\Big), \text{ subject to } \| H_{i\ast}\|_2 = 1 \text{ for all } 1\leq i \leq n.\label{eq:relax0}
\end{equation}
\begin{figure}[!htb]
\centering
\begin{tabular}{cc}
  Laplacian-based ellipsoidal embedding & Laplacian-based spectral embedding \\
\includegraphics[scale = 0.5]{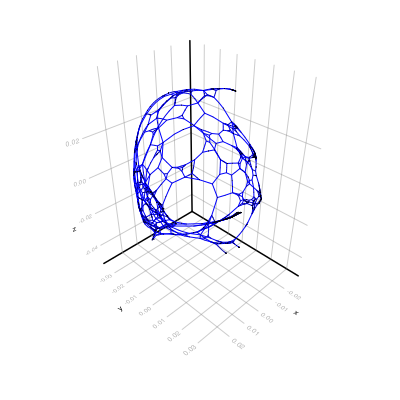} 
&
\includegraphics[scale = 0.5]{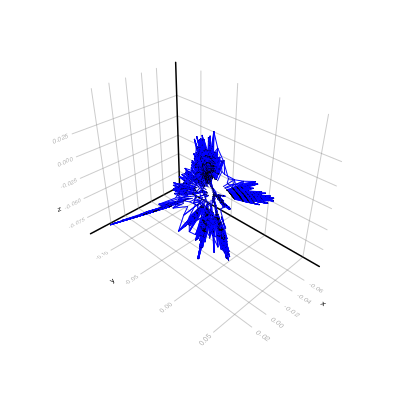}\\
\includegraphics[scale = 0.45]{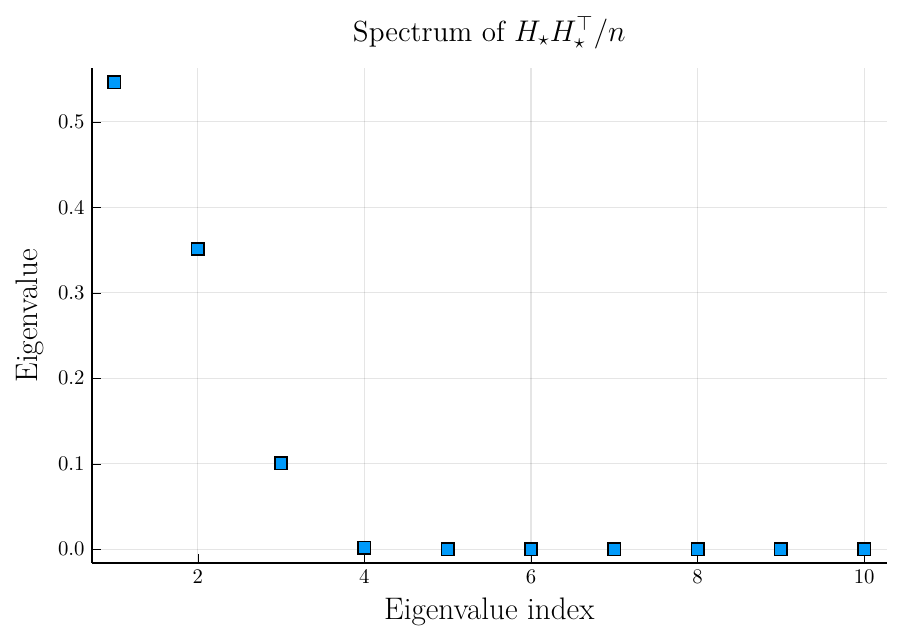}&
\includegraphics[scale = 0.45]{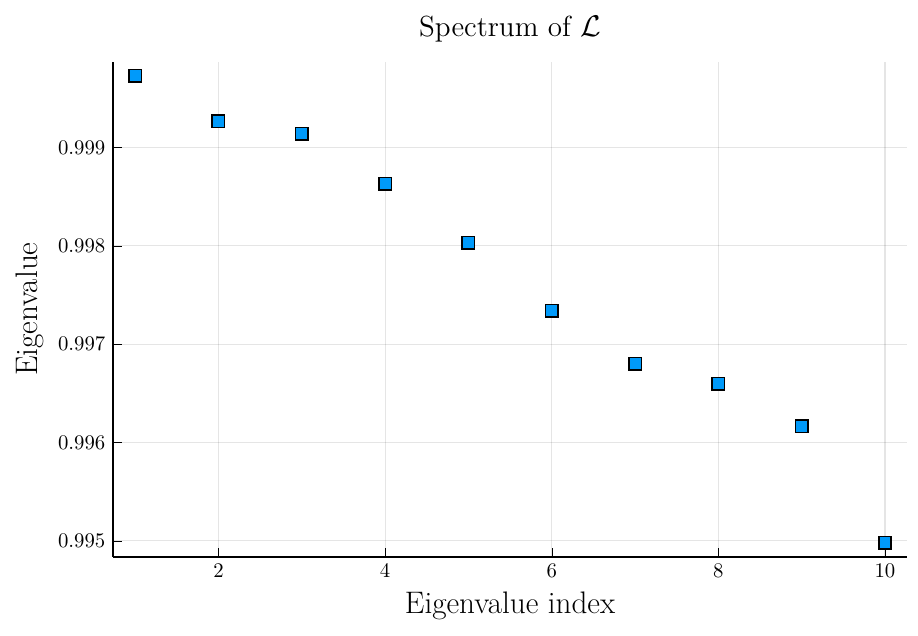}
\end{tabular}
\caption{\textbf{Ellipsoidal vs spectral embedding.} Laplacian-based ellipsoidal embedding (top left, $d_0 = 10$, $M = \mathcal{L}$, see \cref{eq:NormLap}) and the spectral embedding (top right) of the \textsc{PowerEU} graph -- see Table~\ref{Table:RealNetworks} -- thanks to the three leading eigenvectors of the Laplacian $\mathcal{L}$. 
On the bottom left, the spectrum of  $\frac{1}{n}H_\star H^\top_\star$ associated to the embedding with $d_{\rm eff} = 3$. 
The spectrum of this normalized Laplacian matrix $\mathcal{L}$ is given on the botton right.
\label{Fig:Power}}
\end{figure}

Even though the embedding space is \emph{a priori} of dimension $d_0 - 1$ which can be as large as we want, we observe that an optimal solution $H^\star$ of~\eqref{eq:relax0} corresponds to an embedding on a subspace of \emph{lower effective dimension} $d_{\rm eff} \leq d_0 $, that is, the embedding is effectively on $\mathbb{S}^{d_{\rm eff}-1}$. 

As an illustration of the low dimensionality of the embedding in the case of a real networks, the Power Grid of Europe graph (with $n = |\mathcal{V}|=2712$ and $ |\mathcal{E}|=3580$) is embedded in Figure~\ref{Fig:Power}, while a list of other real networks where our method has been applied is given in Table~\ref{Table:RealNetworks} in the appendix.
To show that this embedding is indeed quite different to the often considered spectral embedding based on the same descriptor matrix $M$, the corresponding results are drawn on the right-hand side of Figure~\ref{Fig:Power}. This time in order to illustrate that our approach is not restricted to the modularity, we choose the descriptor matrix to be the normalized Laplacian $\mathcal{L}$ as discussed in the sequel.
As an advantage of the ellipsoidal embedding, we observe on the bottom left of Figure~\ref{Fig:Power} that the effective dimension of the embedding $d_{\rm eff} = 3$ can be read out from the decay of the spectrum, whereas the spectrum of the normalized Laplacian $\mathcal{L}$ does not exhibit a clear gap. The tail of eigenvalues is found to be numerically small and is therefore neglected.
As showed in Figure~\ref{Fig:Foam}, we may expect the ellipsoidal embedding to highlight structures of graphs as it is discussed in the next section.

\subsection{Contributions}

We propose a method for embedding graphs on ellipsoids.
The dimension of this embedding is automatically determined as part of the algorithm and is not required as an input, the analyst merely needs to constrain the maximum dimension $d_0$ of the embedding space for computational purpose.
To obtain our embedding, we make use of a \emph{generalized power method with momentum}, a simple deterministic algorithm with a random initialization.
In practice, this algorithm yields empirically a small effective embedding dimension, highlighting that many graphs can be represented in terms of a low-dimensional parametrization.
For the examples in this paper, the embedding dimensions ranges from $d_{\rm eff} = 2$ for simple graphs to $d_{\rm eff} = 6$ for graphs with more complicated structures. 
For graphs with a large embedding dimension, the embedded data can then be further analysed via multi-dimensional scaling or principal component analysis to visualize the results. In the latter case, the dimensionality reduction relies on a thresholding of components with small variance.

While the mathematical formulations are different, there are certain analogies with spectral embeddings~\cite{QinRohe,Luxburg2007, Rohe2011, RDPgraphs} that we discuss.
As a case study to demonstrate the utility of the derived embedding, we show how the embedding can be utilized to perform graph clustering, while spectral embeddings are commonly used as a preprocessing step for clustering.
For this case study, we perform an embedding in which the modularity matrix -- a well-known tool for graph clustering -- is chosen as the descriptor matrix, and then use the resulting embedding to perform graph clustering. 
Note, however, that any other descriptor matrix could have been chosen in this context resulting in a different embedding.
Given a descriptor matrix, the corresponding (meta)-algorithm consists in first computing the embedding and then performing a clustering procedure, by using a vector partionning algorithm, which is of independent interest.
We call this strategy \emph{embed-and-partition}.
A point of practical interest in this context is that our embedding-based clustering method directly provides an estimate for the relevant number of clusters. Basically, the clustering starts with a given number of `centroid' nodes initially sampled with respect to their degree and the algorithm then optimizes their positions by possibly merging a few of them.
This is in contrast to many other methods based on embeddings, which typically rely on specifying a desired number of clusters a priori and then perform a $k$-means clustering or a similar  procedure on the obtained embedding coordinates.
Interestingly, for a modularity based embedding as studied here, its performance is often comparable to the highly successful Louvain method~\cite{Louvain} on benchmarks graphs in terms of efficiency and for a limited increase of computing time; see \cref{Fig:nmi_vs_mu}.
As a limitation of our {embed-and-partition} approach, whereas $d_{\rm eff}$ is often smaller than $d_0$ as we mentioned above, we observe that the clustering performance is especially good when applied on the $d_0$-dimensional embedding coordinates rather than on the $(d_{\rm eff} - 1)$-dimensional sphere or ellipsoid.

We further present empirical evidence that for networks with modular structure, the weaker the community structure is defined, the higher the (automatically inferred) embedding dimension --- thus highlighting a connection between the hardness of community detection and the ability to compress the network structure via a low-dimensional ellipsoidal embedding.

A code implementing our embedding and partitioning methods is available at \url{https://github.com/mrfanuel/EllipsoidalGraphEmbedding.jl} in the form of a Julia module.

\subsection{Outline}
The remainder of this paper is organized as follows.
To set the scene and provide additional motivation for our embedding, we first provide some concrete examples of trace optimization problems in the context of network analysis in Section~\ref{sec:Motivation}.
Subsequently, the mathematical formulation of the embedding problem is explained in Section~\ref{sec:Math}, and a novel optimization method is proposed to solve it. Some mathematical properties of the algorithms used for obtaining the ellipsoidal embedding are given in Appendix~\ref{sec:AppendixMath}.
We discuss relations and differences of the here proposed embedding with other problem formulations in network analysis and spectral embeddings.
In Section~\ref{sec:Community}, we discuss the relationship of ellipsoidal embeddings to community detection in the form of modularity maximization, as a concrete application for our embedding.
Specifically, we introduce  a greedy algorithm based on the ellipsoidal embedding to obtain the clustering of a graph.
Then, we discuss the results of the embedding based community detection algorithm, by using synthetic and real-world networks with up to one million nodes and several million edges.
Our results show that the proposed embedding faithfully captures relevant structural features of a graph. The details of the numerical simulations are given in  Appendix~\ref{sec:NumericalResults}. 
We conclude with a brief discussion in Section~\ref{sec:Conclusions}.
To improve the readability of the paper, the proofs of the mathematical results are relegated to the Appendix. 
\subsection{Notation}
In terms of notation, we denote by $\mathcal{G}$ a connected graph with vertex set $\mathcal{V}$ and edge set $\mathcal{E}$. 
The number of nodes in the graph is denoted by $n = |\mathcal{V}|$ and we suppose that the graphs under consideration are undirected.
For convenience, we always identify the node-set of a graph with the natural numbers $\{1,\dots,n\}$. 
The adjacency matrix $A\in \mathbb{R}^{n\times n}$ is then defined such that $A_{ij} = 1$ if and only if $i$ is connected to $j$ and $A_{ij} =0$ otherwise.  
It is customary to introduce the degree vector $\degree$ with elements $\degree_i = \sum_{j=1}^{n}A_{ij}$ and total edge weight $m = \sum_{i=1}^{n}\degree_i/2$.  
Based on the degree matrix $D = \diag(\degree)$, we further define the combinatorial Laplacian as $L = D-A$, and the normalized Laplacian $L_N = D^{-1/2}LD^{-1/2}$.
If $M$ is a square, positive semi-definite (psd) matrix we write $M\succeq 0$ (recall that $M$ is \emph{psd} if and only if $v^\top Mv\geq 0$ for all vectors $v$). The $\ell$-th column and $i$-th row of a matrix $H$ will be denoted by $H_{\ast\ell}$ and $H_{i\ast}$ respectively. 
Finally, the nuclear norm of a matrix $M$ is defined as follows: $\|M\|_\star = \Tr(\sqrt{M^\top M})$.

\section{Motivation and Background}\label{sec:Motivation}

\subsection{Ellipsoidal embedding and graph structures}\label{sec:graph_structures}
While graph drawings can help to better understand the topological structure of graphs, we may want to invert this process and try to reason about the graph itself by means of a carefully defined geometrical embedding of the graph.
Hence, we may want to design a graph embedding that reflects certain topological properties of the graph geometrically.
These are natural ideas, and going back and forth between a geometric embedding and a graph representation of data (sometimes implicitly) underpins a host of successful methods in data science.
\begin{itemize}[leftmargin=*]
    \item[-] For instance, spectral clustering~\cite{ShiMalik,Newman2013} may be seen as a graph embedding into a Euclidean space on which we perform `classical' clustering afterwards (e.g., using $k$-means). 
        Many different spectral clustering methods exist that rely on choosing the `right' algebraic representation of the graph, such that certain features of the graph are emphasized~\cite{Luxburg2007,Luxburg2008,Saerens2004,zhang2008}.
    \item[-] Manifold learning techniques such as diffusion maps~\cite{Coifman2005,Lafon2006,Nadler2006} provide another interesting example.
        In this case, we start with a geometric point cloud, from which we construct a graph based on the geometric data.
        From this graph we then derive a new geometric representation of the data by embedding this graph via a set of diffusion coordinates, thereby providing a data parametrization in a lower dimensional Euclidean space.
    \item[-] In generative models for graphs we often posit the existence of a `correct' embedding in the construction, e.g., hyperbolic embeddings and hyperbolic latent models have been proposed to model and fit networks~\cite{Asta2015}.
    Similarly, random dot-product graphs~\cite{RDPgraphs} and other continuous latent graph models~\cite{Lovasz} posit that an observed graph has been generated with an implicit set of latent geometric coordinates.
    \item[-] Even for discrete latent random variable models such as the stochastic block-model~\cite{Abbe}, continuous embeddings provided by the spectral properties of the observed graph can provide useful information about the graph, which is harnessed in spectral methods for community detection such as~\cite{Rohe2011}.
    \item[-] Recent approaches for graph clustering, see \cite{PHR-D22} and 
    \cite{GvdHL21}, also use specifically a spherical embedding.
\end{itemize}
What is common among all these approaches is that an embedding of a graph into a metric space provides us with additional means to approximately solve hard problems, such as graph comparisons, clustering, etc.--- by using the rich toolkit of continuous mathematics within the embedding domain.
Indeed, there is a recent surge of interest in graph embeddings because of this reason: some recent works propose to use machine-learning techniques to learn an embedding to reflect certain topological features of the nodes~\cite{Gutierrez,Hamilton2017,node2vec}.

Typically, a spectral embedding uses the $d_0$ leading eigenvectors of a symmetric $n\times n$ matrix $M$ for $n\geq d_0$. These eigenvectors are columns of $H_\star\in \mathbb{R}^{n\times d_0}$ which is a solution of the following trace maximization problem
\begin{align}
    \maximize_{H\in \mathbb{R}^{n\times d_0}}  \Tr\left( H^\top M H\right)
    \text{ subject to } & H^\top H = \mathbb{I}_{d_0\times d_0},\label{eq:Eigenmodularity}
\end{align}
where the constraint implements the orthonormality of the columns of $H$.
Similar to the spectral embeddings, the ellipsoidal embedding, calculated with~\eqref{eq:relax0}, naturally emerges from trace maximisation problems of the form
\begin{subequations}\label{eq:trace_opt}
\begin{align}
    \maximize_H \quad& \Tr\left( H^\top M H\right) \label{eq:Objective}\\
    \text{subject to}\quad & H \in \mathcal Z\label{eq:Constraint},
\end{align}
\end{subequations}
where $\mathcal Z$ denotes the set of constraints.
Here, the matrix $M$ in \eqref{eq:Objective} is an algebraic descriptor of the network. 
For instance, $M$ could be a Laplacian matrix, or a feature matrix derived from the network such as a matrix $M$ with entry $M_{ij}$ counting all walks up to length $k$ between any two nodes $i,j$.
In order to illustrate the relevance of problems of the type~\eqref{eq:trace_opt}, we give here a few examples of such problems in the context of the analysis of graphs and networks.

\subsection*{Trace optimization problems in network analysis}
Laplacian matrices play a major role in network analysis, as their spectral properties are intimately related to the network structure. 
Accordingly, they have been analysed from a variety of angles~\cite{Chung:1997,Mohar91}.
One fundamental problem in which the graph Laplacian emerges is the problem of graph partitioning.
This problem can be phrased as a \emph{penalized cut} problem~\cite{zhang2008}, which includes other popular notions such as normalized cut~\cite{ShiMalik} and ratio cut~\cite{Chan1994}.

Let $H\in \{0,1\}^{n\times k}$ be a binary indicator matrix associated to a partition of the graph with $k$ clusters, i.e., $H_{ic}=1$ if $i$ belongs to group $c$ and zero otherwise.
Based on this definition, the penalized cut problem is to minimize the objective function
\begin{align*}
\Tr\left( H^\top L H (H^\top \Delta H)^{-1}\right) \quad \text{(Penalized Cut)},
\end{align*}
subject to the constraint that $H$ is a binary indicator matrix of the form described above,  and $\Delta$ is a positive definite diagonal weighting matrix.
Note that this objective may be rewritten in the following more compact form
$\Tr \left( Z^\top L Z \right)$ with $Z = H(H^\top \Pi H)^{-1/2}$, i.e., can be directly mapped to problem~\eqref{eq:trace_opt}.

Apart from penalized cut, this class of optimization problems includes many other problems of interest.
For instance, the  formulation~\eqref{eq:trace_opt} includes maximum likelihood estimation of the partitions of certain stochastic blockmodels~\cite{Peixoto,Amini2018,Hajek2016}, a type of generative network models that has gained enormous interest in network analysis recently. 
Furthermore, several synchronization problems~\cite{MontanariPNAS,Boumal} can be formulated in this form, such as the $\Uone$-synchronization problem on graphs~\cite{Singer2011,Boumal,Boumal:2016}:
\begin{align*}
\max_{H\in \mathbb{C}^n}  \Tr(\Theta H H^*) \text{ subject to } H \in \Uone^n, \quad \text{($\Uone$-synchronization)},
\end{align*}
where $\Theta$ is a Hermitian matrix, with entries such that $\Theta_{ij} = \exp \left(\rmi \theta_{ij}\right)$ if $ij$ is an edge of the graph and $\Theta_{ij} =0$ otherwise.  
Here, $H^*$ denotes the Hermitian conjugate of $H$.

Another class of important problems of the above form are those associated to modularity optimization, which we will adopt as our running example in the following.
We remark, however, that most of our arguments are equally applicable, \emph{mutatis mutandis}, to other problem contexts.

\subsection*{Modularity maximization}\label{subsec:modularity}
For a given network with adjacency matrix $A$, let $Q$ be the modularity matrix given by 
\begin{equation}
    Q = \frac{1}{2m}\left(A-\frac{\degree \degree^\top }{2m} \right).\label{e:Q_matrix}
\end{equation}
The problem of optimizing the modularity can be cast in the form~\eqref{eq:trace_opt} as 
\begin{subequations}\label{eq:Combinatorial}
\begin{align}
    \maximize_H\quad  & \Tr\Big(H^\top Q H\Big), \quad \text{(modularity maximization)}\\
    \text{subject to}\quad & H \in \mathcal Z,
\end{align}
\end{subequations}
where $\mathcal Z$ is the set of partition indicator matrices with any number of groups $k$, which obey the definition above: each node is in one and only one group and $H_{ic}=1$ if $i$ belongs to group $c$ and zero otherwise.
Similar to the other problems discussed above, modularity optimization is an NP-hard problem~\cite{Brandes2006}, and accordingly several heuristics have been proposed to solve the above problem, including greedy~\cite{Louvain} and spectral algorithms~\cite{NewmanSpectral}.


\section{Ellipsoidal embeddings\label{sec:Math}}
In view of the embedding interpretation of trace-optimization problems such as modularity optimization, we propose here another embedding that consists in finding a generalized label matrix $H\in\mathbb{R}^{n\times d_0}$ whose rows are to be interpreted as coordinate vectors, by solving~\eqref{eq:relax0} with $M = Q$.
Observe that in contrast to the spectral embedding~\eqref{eq:Eigenmodularity} formulation, where the \emph{columns} of $H$ were supposed to have unit 2-norm, we here apply a constraint on the \emph{rows} of $H$.
The resulting formulation~\eqref{eq:relax0} is a relaxation of~\eqref{eq:Combinatorial} in that every matrix $H \in \mathcal Z$ fulfills the condition $\| H_{i\ast}\|_2 = 1$.
To better understand the above optimization problem, and how it relates to an ellipsoidal (spherical) embedding let us comment on a few features of the above formulation.
First, notice that the enforced constraints on $H$ imply that the feasible space for $H$ is a Cartesian product of spheres $H\in (\mathbb{S}^{d_0-1})^n$, i.e., every row of $H$ defines a point on a hypersphere of dimension $d_0-1$.
Hence the matrix $H$ defines embedding coordinates for each node in the graph, that can be interpreted as points on a hypersphere.
Note that the resulting set of coordinates is only unique up to unitary transformation.
Namely, for any orthogonal matrix $U$, i.e. satisfying $U^\top U = I = U U^\top$, the matrices $H$ and $HU$ will have exactly the same objective value.
Only the matrix $HH^\top$ is invariant under these orthogonal transformations.

Unlike in the spectral case~\eqref{eq:Eigenmodularity}, $d_0$ does not correspond directly to the embedding dimension, but rather corresponds to an upper-bound of the embedding dimension. Interestingly, in many cases, the optimal embedding can have several `empty' columns in $H$, which can be dropped without loss of information.
In practice, in our simulations, we often chose the integer $d_0$ to be smaller than $50$ and never larger than $250$, and we observe that the obtained embedding dimension (see the following sections for a more detailed discussion) is typically much smaller than $d_0$. This is the empirical reason why the ellipsoidal embedding is often low dimensional.

\subsection{Defining spherical and ellipsoidal embeddings}
Let $H_\star$ be an optimal solution of the embedding problem \eqref{eq:relax0} and let $d_{\rm eff} = \rank(H_\star)\leq d_0$. 
For each node $i$, the $i$-th row of $H_\star$ now defines an embedding of the node in a sphere.
However, as this solution is only unique up to rotations/reflections.
Hence, aiming to reduce embedding invariances,  we employ a singular value decomposition (SVD) for $H$.

\begin{definition}[Spherical and Ellipsoidal embeddings]\label{Def:Embedding} Let $H_\star = U S V^\top$ be a SVD of a solution of~(\ref{eq:relax0}) with $S = \diag(s)$, where $s_1 \ge \ldots \ge s_r > 0$. Further, let $U_{i\ast}$ and $\Sigma_{i\ast}$ be the $i$-th row of $U$ and $\Sigma := US$, respectively.
We define the spherical embedding by the map
\begin{equation*}
i\mapsto \Sigma_{i\ast}:= \begin{bmatrix}
s_1 U_{i1} & s_2 U_{i2} &\dots & s_r U_{ir}
\end{bmatrix},\quad 1\leq i\leq n,
\end{equation*} whereas the
ellipsoidal embedding is defined as
\begin{equation*}
i\mapsto U_{i\ast} := \begin{bmatrix}
U_{i1} & U_{i2} &\dots & U_{ir}
\end{bmatrix},\quad 1\leq i\leq n.
 \end{equation*}
\end{definition}

To see that the above mappings define an embedding on a sphere in the same way as the rows of $H_\star$, consider the matrix $\rho = H_\star H^\top_\star = \Sigma\Sigma^\top \in \mathbb{R}^{n\times n}$, which is a \emph{psd} matrix with elements given by the inner product $\rho_{ij} = \Sigma_{i\ast}\Sigma_{j\ast}^\top$. 
Since by definition of the embedding problem~\eqref{eq:relax0}, the diagonal elements of $\rho$ have to be equal to 1, we know that $\|\Sigma_{i\ast}\|_2 = 1$, and hence the embedding vector $\Sigma_{i\ast}$ defines a point on the unit sphere $\mathbb{S}^{r-1}$. 

To understand the ellipsoidal embedding, we can define an alternative inner product on $\mathbb{R}^{r}$ denoted by $\langle\cdot, \cdot\rangle_{S^2}$ based on  the diagonal positive definite matrix $S^{2}$, where $S$ is the matrix of singular values from the SVD of $H$. 
In terms of this inner product the element $\rho$ can be reinterpreted as $\rho_{ij}=\langle U_{i\ast}, U_{j\ast}\rangle_{S^2} = U_{i\ast} S^2 U_{j\ast}^\top$. 
Hence, we see that each $U_{i\ast}$ belongs to an ellipsoid in $\mathbb{R}^{r}$ determined by the equation $\langle u, u\rangle_{S^2} = 1$.

The singular value decomposition of $H_\star$ -- closely related to the spectral decomposition $\rho = H_\star H^\top_\star=\Sigma\Sigma^\top$ -- further provides us with a simple estimate of the effective dimension of the embedding.
We define the effective dimension as:
\begin{equation}
 d_{\rm eff}(\epsilon) = \min\Big\{1\leq r \leq d_0\Big|\sum_{\ell =1}^{r}s^2_\ell(\rho)>(1-\epsilon) \times \Tr(\rho)\Big\}.\label{eq:deff}
 \end{equation}
 In this paper, we choose $\epsilon = 0.01$.
Intuitively, the above definition discounts eigen-coordinates which contribute less than $1\%$ to the total variation in the embedding coordinates.
The specific value of $\epsilon$ may here be interpreted as a `significance' value, which can be chosen by the analyst.

\begin{myremark}[Eigenvalue thresholding]
    If $H_\star$ is a solution of the embedding problem, then the effective embedding $H_\text{eff}$  corresponds to a truncation of the  SVD of $H_\star$ to its $d_{\rm eff}$ largest singular values. 
    The nuclear norms of the invariants $\rho =H_\star H^\top_\star$ and $\rho_{\rm eff} =H_{\rm eff}H_{\rm eff}^\top$ are related by
$
\|\rho-\rho_{\rm eff}\|_{\star}\leq \epsilon \|\rho\|_{\star},
$
where $\|\rho\|_{\star}= n$.
As we choose here $\epsilon = 0.01$, this means that the relative error (as measured by the nuclear norm) between the effective embedding and the optimal embedded is less than $1\%$.
\end{myremark}
\begin{myremark}[Orientation ambiguity]\label{rem:Orientation}
The  embedding coordinates given in Definition~\ref{Def:Embedding} are ordered according to the magnitude of the singular values. 
Provided that all singular values are distinct, this implies that there is no ambiguity in terms of the ordering of the coordinates.
There remains one source of ambiguity, however, namely, a direction change of a coordinate axis. 
\end{myremark}

Note that the ambiguity discussed in Remark~\ref{rem:Orientation} is also encountered in spectral embeddings and is essentially unavoidable due to the symmetry of the problem.
In the context of spectral embeddings, the above ambiguity corresponds to the fact that any (unit) eigenvector is only defined up to a phase.
In practice, these issues of non-uniqueness can be ignored for most applications:
typically, we are interested in the relative positions of the nodes, rather than their absolute positions in the embedding space.

\subsection{Computing ellipsoidal embeddings}
The embedding problem~\eqref{eq:relax0} can be solved in a number of different ways.
In this work we employ a generalized power method, as described in Algorithm~\ref{Alg1}, that is inspired from~\cite{Journee:2010:GPM:1756006.1756021,Boumal,ChenCandes,MasterThesis}.

To see how the method works, first notice that the diagonal of the descriptor matrix, here illustrated by the modularity matrix $Q$ does not influence the solution of~\eqref{eq:relax0} but merely shifts the objective value by a constant. 
To see this, observe that for any diagonal matrix $D$, we have $\Tr(H^\top (Q - D) H) = \Tr(H^\top Q H) - \sum_i D_{ii}$.
Since in general $Q$ may be indefinite, we thus employ a preprocessing step to make any descriptor matrix positive definite, by simply shifting the spectrum of $Q$ with a diagonal matrix.
Specifically we apply the transformation $ Q\mapsto  K =Q+\diag(v)$.
Here $v$ is chosen such $K$ is strictly diagonally dominant and therefore $K\succ 0$. 
Specifically, here we define
\begin{equation}
K_{ij}=
\begin{cases}
 Q_{ij} & \text{ if } i\neq j\\
 1+\sum_{k\neq i}|Q_{ik}| & \text{ if } i=j
 \end{cases}.\label{eq:Preprocess}
\end{equation}

\begin{algorithm}[tb!]
\caption{Generalized power method \cite{Journee:2010:GPM:1756006.1756021,Boumal,ChenCandes}\label{Alg1}}
\begin{algorithmic}[1]
\Require Symmetric positive definite matrix $K\in \mathbb{R}^{n\times n}$;
and an initial embedding $x_0\in \mathbb{R}^{n\times d_0}$ such that $\Pi(x_0) = x_0$; see \eqref{eq:P}. Fix $0<\text{tol}<1$, and set $m=0$.
\State {\bf do}
\State $m = m+1$
\State $x_m = \Pi(Kx_{m-1})$
\State $o_{m} = \Tr(x_{m}^{\top} K x_{m})$
\State \textbf{while} ($m\le 1$ or $|o_{m} -o_{m-1}|/o_{m-1} \ge {\rm tol}$)
\end{algorithmic}
\end{algorithm}

Using this shifted descriptor matrix, Algorithm~\ref{Alg1} now solves the embedding problem iteratively, starting from an initially feasible solution.
Inspired from~\cite{ChenCandes}, here the initialization is obtained by selecting uniformly at random $d_0$ columns of $K$ and then by projecting the resulting $n\times d_0$ matrix on the product of spheres $ (\mathbb{S}^{d_0-1})^{n}$ by using the projection operator $\Pi$, which maps any matrix $H\in \mathbb{R}^{n\times d_0}$ such that $H_{i\ast}\neq 0$ for all $i=1,\dots,n$ onto a spherical embedding matrix, by normalizing the rows of $H$, namely
\begin{equation}
\Pi(H)_{i\ast} = H_{i\ast}/\|H_{i\ast}\|_2,\text{ for all } 1\leq i\leq n,\label{eq:P}
\end{equation}
where we recall that $H_{i\ast}$ is the $i$-row of $H$.

Starting from a feasible initial condition, we alternate between applying our (shifted) descriptor matrix to the current embedding, and then project the result back again onto a hypersphere.
Intuitively, the repeated multiplication of $K$ aligns the current iterate with the dominant subspace of $K$ akin to a power-method for eigensolvers, thereby increasing the objective value, while the projection step acts as a normalization step and ensures that we maintain feasibility. 
The algorithm is stopped when the relative variation of consecutive objectives does not exceed a particular tolerance.
In this paper we choose the tolerance $\mathrm{tol} = 1\mathrm{e}{-08}$ unless stated otherwise.

Mathematically, Proposition~\ref{PropositionMonotoneSimplified} gives a lower bound on the improvement between successive objectives values.
\begin{myprop}\label{PropositionMonotoneSimplified}
Let $K\in \mathbb{R}^{n\times n}$  symmetric such that $|K_{ii}|> 1+\sum_{k\neq i}|K_{ik}|$ for all $1\leq i\leq n$. Let the objective function be $f(x) = \Tr(x^\top K x)$. Then, the sequence of objectives for the iteration $x_{m+1} = \Pi(Kx_{m})$  satisfies
\[
 f(x_{m+1})-f(x_m)> \|x_{m+1}-x_m\|_2^{2}
\]
for all $m\geq 0$.
\end{myprop}
The proof of a more general version of Proposition~\ref{PropositionMonotoneSimplified} can be found in Appendix~\ref{appendix:GPM}.

While Algorithm~\ref{Alg1} provides us with a practical algorithm to solve our embedding problem, in order to speed up the optimization, we propose Algorithm~\ref{Alg2}, which includes a `momentum' term to accelerate the iterations~\cite{Nesterov1983}.
The advantage of this preprocessing for the convergence of the Generalized Power Method (GPM) (Algorithm~\ref{Alg1}) as well as a more detailed theoretical analysis of both algorithms is discussed in the Appendix~\ref{appendix:GPM}.

\begin{algorithm}[tb!]
\caption{Generalized power method with momentum\label{Alg2}}
\begin{algorithmic}[1]
\Require Symmetric positive definite matrix $K\in \mathbb{R}^{n\times n}$;
and an initial $x_0\in \mathbb{R}^{n\times d_0}$ such that $\Pi(x_0) = x_0$. Initialize $y_0 = Kx_0$.
Fix $ 0<{\rm tol}<1$ and set $n=0$.\\
Let $r_n$ be defined as $r_n := (n-1)/(n+2)$ for $n\in\{1,2,\dots\}$.  
\State {\bf do}
\State $n=n+1$
\State $o_{n-1} = \Tr(y_{n-1}^{\top} x_{n-1})$,
\State  $y_{n} =K x_{n-1} $,
\State  $x_{n} = \Pi( y_{n}+r_n (y_{n} - y_{n-1}))$,
\State \textbf{while} ($n\le 1$ or $|o_{n-1} -o_{n-2}|/o_{n-2} \ge {\rm tol}$)
\end{algorithmic}
\end{algorithm} 

\subsection{Effective embedding dimension}
As announced above, the embedding dimension is often low.
To gain some further insight into this empirical fact, let us introduce a closely related SDP:
\begin{equation}
\maximize\Tr(\rho K ) \text{ subject to } \rho\succeq 0 \text{ and }\rho_{ii} = 1 \text{ for all } 1\leq i\leq n\label{eq:SDP},
\end{equation}
where we defined the square \emph{psd} matrix $\rho = HH^\top$.

Although we will not numerically solve this SDP, it can be shown that a solution of the first order optimality condition of the embedding problem~\eqref{eq:relax0} also satisfies the complementary slackness condition of~\eqref{eq:SDP}; see Appendix~\ref{appendix:GPM} for more details.
Indeed, under certain circumstances the maximum of both problems correspond~\cite{Boumal:2016}, i.e., the non-convex embedding problem~\eqref{eq:relax0} can be effectively solved (up to rotations) by the convex program~\eqref{eq:SDP}.
We summarize these results in Proposition~\ref{Prop:Rank}.
To write our results compactly, here we use $\ddiag(M)$ to denote the diagonal matrix obtained by replacing all off-diagonal elements of $M$ by zero.
\begin{myprop}[Equivalence with a nuclear norm minimization]\label{Prop:Rank}
Let $K\in \mathbb{R}^{n\times n}$ be a \emph{psd} matrix with a maximal eigenvalue strictly smaller than $\lambda>0$ and let $\Sigma \in \mathbb{R}^{n\times n}$ be the invertible matrix with orthogonal rows such that $\lambda\mathbb{I}-K= \Sigma\Sigma^\top$. Then, the optimal solution $X^\star$ of
\begin{equation*}
\minimize_{X\succeq 0}\|X\|_{\star} ,\text{ subject to}\ \ddiag\Big( (\Sigma^{-1})^\top X\Sigma^{-1}\Big) = \ddiag(K),
\end{equation*}
has the same rank as the optimal solution of~(\ref{eq:relax0}) $\rho^\star$ and is given by $X^\star = \Sigma^\top \rho^\star\Sigma$. 
\end{myprop}
Proposition~\ref{Prop:Rank} is completely analogous to Proposition 3.1 of \cite{SDPEmbedding} where a proof is given.
Observe that in view of Proposition~\ref{Prop:Rank} the problem~\eqref{eq:SDP}, and thus our related embedding problem~\eqref{eq:relax0} is equivalent to a nuclear norm minimization subject to linear constraints which promotes a low-rank solution and thus a low embedding dimension.
Indeed, the minimization of the nuclear norm is a relaxation of the minimization of the rank of a matrix.

\subsection{Using other descriptor matrices to derive ellipsoidal embeddings}

Our description of the embedding so far has used the modularity matrix as our primary example of a feature matrix.
While modularity is certainly one of the most well-known feature matrices related to network analysis, there is nothing about our problem formulation that forces us to stick to modularity.
Indeed, is worth remarking again that we may also us alternative matrices to derive alternative embeddings with different interpretations.
The basic requirements on the descriptor matrix $M$ are that it is symmetric and contains both positive and negative entries.
For instance, the modularity matrix is designed to detect an assortative group structure in a network, i.e., it is expected to emphasize groups of nodes which are densely connected with each other. 
However, if we are interested in disassortative structures (e.g., bipartite structure), we may want to consider a descriptor based on the squared adjacency matrix $A^2$. For instance, we could simply consider a modularity matrix derived from the network with adjacency $A^2$.
Another choice for a descriptor matrix is given by a form of Laplacian matrix.
Let $D = \diag(\rmd)$ be the diagonal degree matrix and let $\pi_i = \rmd_i/\sum_j \rmd_j$. 
An embedding related to the normalized Laplacian matrix may then be defined via the following descriptor matrix\footnote{Strictly speaking, the matrix $\mathcal L$ and the normalized Laplacian $L_N$ are different matrices. However, note that $\mathcal L = I - L_N -\sqrt{\pi}\sqrt{\pi}^{\top}$ simply corresponds to a shifted version of $L_N$ with a rank-1 correction term.}
\begin{equation}
\mathcal L =D^{-1/2}AD^{-1/2}-\sqrt{\pi}\sqrt{\pi}^{\top}.\label{eq:NormLap}
\end{equation}
A final choice for a descriptor matrix is given by autocovariance matrix of a random walk on the graph, as it features in the Markov stability framework~\cite{Delvenne12755,Schaub2012b,Delvenne2013}, which allows to sweep the graph structures at different scales.
More generally, we may consider descriptor matrices derived from more general dynamical (covariance) kernels~\cite{Schaub2018a}, in order to capture certain dynamical features of the problem at hand.
We postpone the study of these alternatives for a further work.

\section{Case Study: Ellipsoidal embeddings for graph partitioning\label{sec:Community}}
One task in network analysis that has enjoyed tremendous interest over the past decades is community detection -- the task of partitioning a network into groups of nodes according to some pre-specified criterion.
In the following, we show how we can use our spherical embedding to perform community detection for networks.
While the resulting algorithm may be seen as an independent non-parametric community detection (meta-)heuristics (depending on the chosen descriptor matrix) in its own right, our goal here is primarily to illustrate the utility of the embedding using this task as a case study.

For simplicity, we will use again the modularity matrix as descriptor matrix of our embedding here.
Note that we do not aim to optimize modularity here directly, nor do we advocate modularity optimization as the method of choice for community detection.
However, choosing a modularity based embedding enables us to relate the resulting clustering to the large literature of methods for modularity optimization and thus provides some form of external validation for the utility of the embedding.
For comparison we therefore also computed network clusterings according to the Louvain method~\cite{Louvain}, which is known to perform well for modularity optimization on large graphs.

\subsection{Embed-and-partition}
In order to find clusters in the embedding, we take inspiration from the well-known $k$-means algorithm and the vector partitioning methods proposed in~\cite{ZhangNewman,LiuBarahona}.
Let $H\in \{0,1\}^{n\times k}$ be a binary membership matrix associated to a partition of the graph with $k$ clusters, i.e., each node is in one and only one cluster, and $H_{ic}=1$ if  $i$ is in the cluster $c$ and zero otherwise. Also, we denote by $c_i\in\{1,\dots, k\}$ the cluster index of $i\in\{1,\dots, n\}$.
Then, we aim to optimize the following objective
\begin{equation}
    \tilde{z} = \Tr\left(H^\top Z H\right)  = \sum_{\ell=1}^{k} \Big\|\sum_{\{i|c_{i}=\ell\}} U_{i\ast}\Big\|_2^{2},\label{eq:ObjectiveVectorPartition}
\end{equation}
where $Z = UU^\top$ in the case of ellipsoidal embedding, see Definition~\ref{Def:Embedding}, and where $H$ is a binary membership matrix.
Like most partitioning problems, the exact maximization of the latter objective function over all binary membership matrices is performed by a greedy approach; see Algorithm~\ref{Alg:VecPart}.
In the spirit of vector partitioning, each of the sum above is associated to a centroid vector $R_\ell = \sum_{\{i|c_{i}=\ell\}} U_{i\ast}.$
Following~\cite{ZhangNewman}, moving node $i$ from community $\ell$ to community $\ell'$  yields the following change in the objective: $\Delta \tilde{z} =  2 U_{i\ast}(R_{\ell'}-R_\ell)^\top-2$. 
This means that if $ U_{i\ast} R_{\ell'}^\top> U_{i\ast} R_\ell^\top +1$, the objective is improved by changing node $i$ from community $\ell$ to community $\ell'$.
This remark motivates the iteration given in Algorithm~\ref{Alg:VecPart}.

In view of these remarks, the partitioning algorithm proceeds as follows.
We first initialize the algorithm according to:
\begin{enumerate}
    \item Draw $k$ centroid vectors $R_1,\dots, R_k$ without replacement from the set of position vectors  $\mathcal U = \{U_{1\ast},\dots,U_{n\ast}\}$ according to the distribution $\pi$.
\item For all $1\leq i\leq n$, calculate 
$c_i^{(0)} \in \argmax_{1\leq\ell\leq k} U_{i\ast} R_\ell^\top$.
\item For all $1\leq\ell\leq k$, compute $
R_\ell^{(0)} = \sum_{\{i|c_{i}^{(0)}=\ell\}} U_{i\ast}$.
\end{enumerate}
Here the probability distribution $\pi$ for the initial sampling is chosen to be proportional to the degree of the node $\pi_i = \rmd_i/\sum_{i}\rmd_i$.
We then iterate over the cluster-assignments and centroid updates in an alternating fashion as outlined in Algorithm~\ref{Alg:VecPart}.
\begin{algorithm}[tb!]
\caption{Vector partitioning~\cite{ZhangNewman} \label{Alg:VecPart}}
\begin{algorithmic}[1]
\Require embedding $\{U_{i\ast}\}_{i=1,\dots,n}$, initial partition $\{c_i^{(0)}\}_{i=1,\dots, n}$ and centroid vectors $\{R_\ell^{(0)}\}_{\ell =1,\dots, k}$.
 \State {\bf do} 
\State Find $c_{i}^{(n)}\in \argmax_{1\leq \ell\leq k}U_{i\ast} R^{(n-1)\top}_\ell$.
\State  Update $R^{(n)}_\ell = \sum_{\{i|c_{i}^{(n)}=\ell\}} U_{i\ast}$.
\State {\bf while} modularity of the partition $\{c_i^{(n)}\}_{i=1,\dots, n}$ keeps increasing.
\end{algorithmic}
\end{algorithm}
In practice, we update the communities as long as the objective $\Tr\left(H^\top Z H\right)$ of the partition associated to $H$ increases, otherwise we stop. 
We notice empirically that if $k$ clusters are initialized at random, then due to the shape of the embedding several clusters will be associated to empty partitions after a few iterations. 
In this case our procedure will yield a number of clusters smaller than or equal to the original supplied upper bound $k$.

As mentioned already above, note that for obtaining partitions with a good modularity value with embed-and-partition in the simulations of this paper, we always define the embedding $U_{i\star}$ of \cref{Def:Embedding} from the untruncated SVD.


\begin{figure}[h]
\centering
\begin{tabular}{cc}
\includegraphics[scale = 0.46]{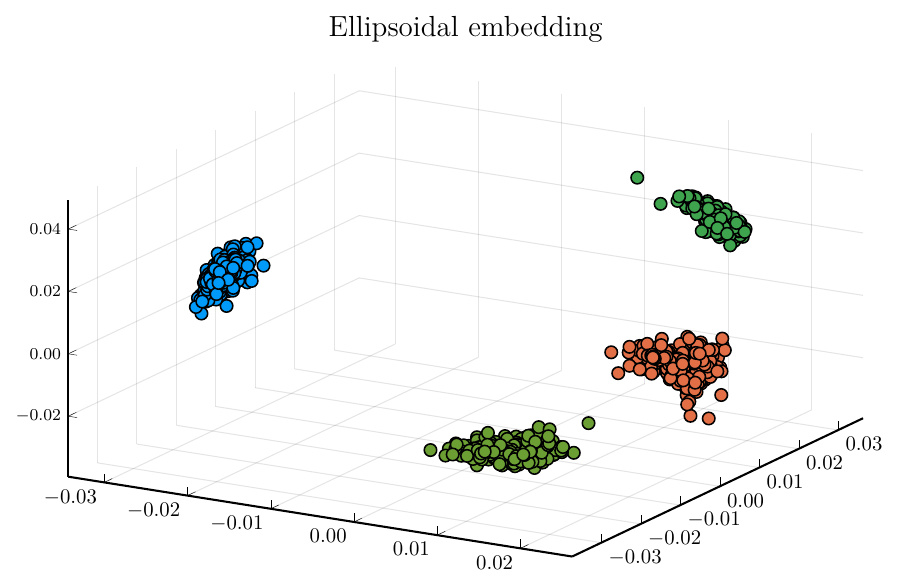}
&\includegraphics[scale = 0.46]{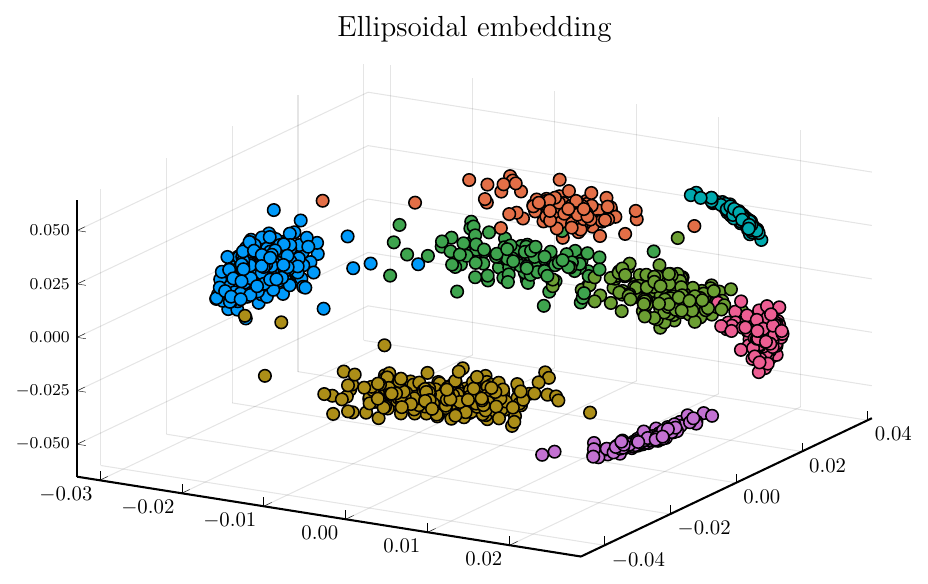}\\
\includegraphics[scale = 0.48]{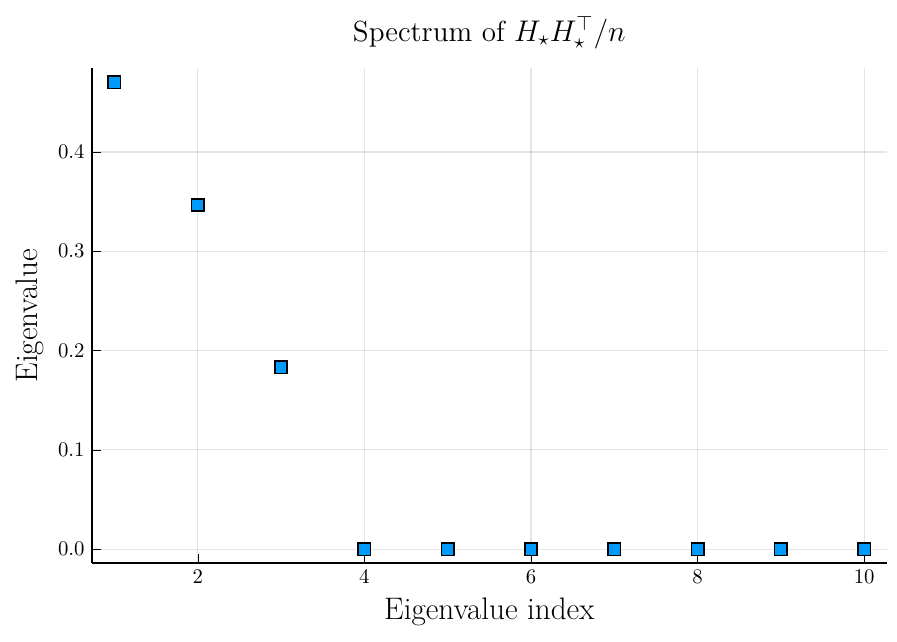}
&\includegraphics[scale = 0.48]{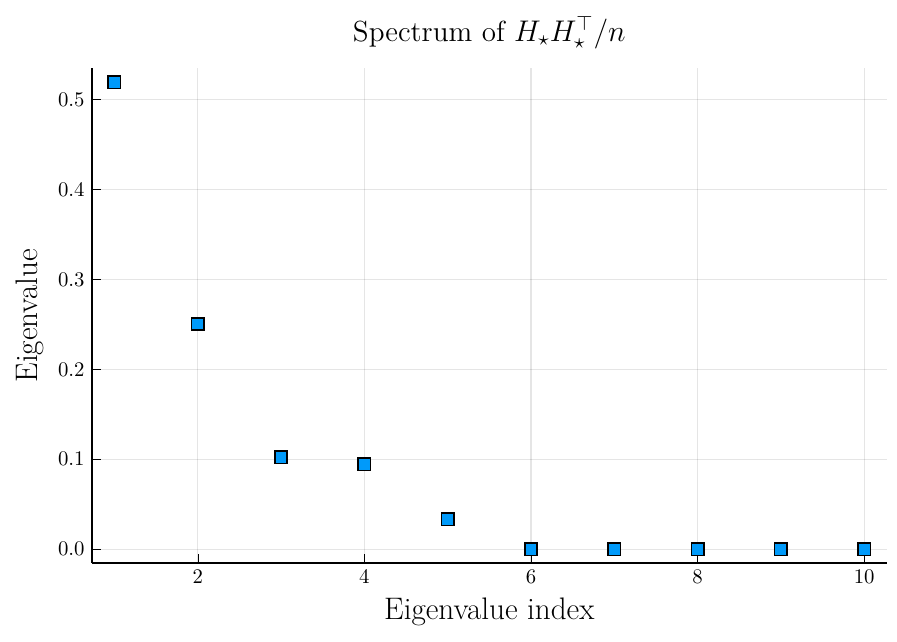}\\
\end{tabular}
\caption{\textbf{Modularity-based ellipsoidal embeddings for graphs with community structure.} A vizualization of the embedding of two LFR benchmark graphs with $2000$ nodes with $d_0=10$; \textsc{LFR1} (left, $4$ planted communities and ${ d}_{\rm eff}=3$) and \textsc{LFR2} (right, $8$ planted communities, a larger mixing parameter and ${ d}_{\rm eff}=5$), see \cref{a:fig_embed_toy} for details. The colors indicate the true community structure. On the bottom, the eigenvalues of $\frac{1}{n}H_\star H^\top_\star$.
Our embed-and-partition retrieves the planted communities in both cases. 
\label{Fig:embed_toy}}
\end{figure}
\subsubsection{Numerical results for embedding based graph partitioning of benchmark graphs}

To perform our synthetic experiments, we created a range of different benchmark graphs using the model of Lancichinetti, Fortunato and Radicchi~(LFR)~\cite{LFBenhmarks}, which simulates graphs with community structures inspired by statistical patterns observed in real-world networks. 
\begin{figure}[tb!]
  \centering
  \includegraphics[scale=0.7]{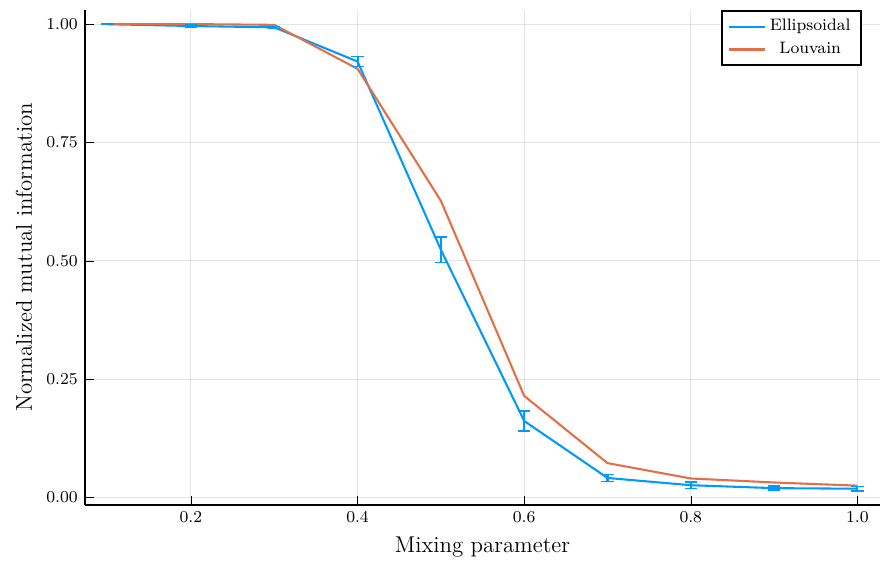}
  \caption{NMI \emph{vs} mixing parameter $\texttt{mu}$ of LFR benchmark graphs with $n=1000$ nodes.
  These graphs were generated with different mixing parameters ranging from $0.1$ to $1$; see \cref{a:nmi_vs_mixing} for the numerical setting.
  An ellipsoidal embedding was computed with $d_0 = 30$ and communities were retrieved thanks to \cref{Alg:VecPart} with $k= 100$ initialized centroids.
  The NMI between the planted and retrieved community structure is here displayed as a function of the mixing parameter.
  The whole procedure was repeated independently $3$ times and  averages as well as standard deviations are reported.
  We refer to \cref{Fig:nmi_vs_mu_ablation} for a study of the sensitivity to the choice of $k$ and $d_0$.\label{Fig:nmi_vs_mu}}
\end{figure}

To gain some further intuition of how our method operates in this task, a visualization of the embedding of FLR benchmark graph is shown in Figure~\ref{Fig:embed_toy}. In both cases, the effective dimension of the embedding is indeed small as it may be seen from the spectra at the bottom of Figure~\ref{Fig:embed_toy}, while the planted communities are recovered by our partitioning method.
In a more extensive study, our clustering results are compared in \cref{Fig:nmi_vs_mu} on LFR benchmarks of various mixing parameters with the Louvain method~\cite{Louvain}.
A conclusion that can be drawn from those comparisons is that our method yields competitive partitions in terms of quality.

\subsubsection{Numerical results for real-world graphs}
Several real networks given in Table~\ref{Table:RealNetworks} were also used to compare partitions obtained with modularity-based ellipsoidal embeddings with the following baselines: the Louvain method, node2vec $+$ k-means and a spectral method based on the modularity matrix followed by vector partitioning.
The results can be found in Table~\ref{fig:ResultsReal}. 
We observe that in the case of those real-world networks, our partitioning method often obtains relatively good modularity values.

\section{Conclusions}\label{sec:Conclusions}
Taking inspiration from spectral relaxations of trace optimization problems, we have proposed a general ellipsoidal embedding algorithm for networks.
We have discussed several connections of this approach to spectral clustering and other methods proposed in the literature and provided a simple, efficient algorithm to compute such an embedding.
We have further shown that our embedding can be utilized for community detection by applying a vector partitioning algorithm in the embedding space derived from the modularity matrix, which may be of independent interest.
Interestingly the computed embedding dimension, which can be selected in an automatic fashion, appears to be indicative of the `structural complexity' of the studies network, and can be used as a lower bound for the amount of clusters present in the network.

There are a number of interesting research directions based on this work worth pursuing in future research.
For instance, it would be interesting to characterize the relationship between the network structure and the optimal embedding dimension in more detail.
In particular, while numerically we have observed that the embedding dimension can serve as a robust proxy for the complexity of the network, it would be interesting to see whether this observation can be formalized.
One possible way forward here would be to study a generative model which could be related to a spherical embedding such as the $\mathbb{S}_1$ model~\cite{Serrano2008}, or random dot-product graphs~\cite{RDPgraphs}.
In this context it would also be insightful to understand the relationship to associated spectral embeddings better, which follow a related, yet distinct paradigm.

From the algorithmic perspective, a new preprocessing method for the Generalized Power Method was proposed in this work as well as a new algorithm: the Generalized Power Method with Momentum. Although we have no proof yet of the convergence for this new algorithm, it was showed empirically to converge markedly faster. 
Especially since the proof techniques of the accelerated gradient methods~\cite{Nesterov1983} do not seem to be applicable in our context, we think it is of theoretical interest to study its convergence properties in more detail.

Finally, there are some interesting interpretations of the here proposed method, as discussed in Section~\ref{sec:additional_interpretations}, which will be worth exploring further.
In particular, the connection of the ellipsoidal embeddings to quantum dynamics (density matrices) suggests further investigation.
There has been significant interest recently in quantum dynamics such as random walks on network~\cite{kempe2003quantum} and it would be of interest, e.g., to explore `quantum descriptor matrices' of graphs and their resulting embeddings.

\paragraph{Acknowledgments}
M.F. acknowledges support from ERC grant BLACKJACK (ERC-2019-STG-851866, PI: R. Bardenet).
M.T.S. acknowledges funding from the European Union's Horizon 2020 research and innovation programme under the Marie Sklodowska-Curie grant agreement No 702410 and the Ministry of Culture and Science (MKW) of the German State of North Rhine-Westphalia ("NRW Rückkehrprogramm"). J.-C. D. acknowledges support from the Research Project PDR TheCirco of the National Fund for Scientific Research (F.R.S.-FNRS) of Belgium.

\appendix

\section{Optimization problem and preprocessing\label{sec:AppendixMath}}
The purpose of this section is to state and prove the convergence properties of the Generalized Power Method in the context of the problem addressed in this paper.
The choice of pre-processing used is also discussed. 

Let  $Q\in \mathbb{R}^{n\times n}$ be a symmetric matrix and an integer $d_0>1$. We also introduce the notation $(\mathbb{S}^{d_0-1})^n$ and $(\mathbb{B}^{d_0})^n$ for the Cartesian product of unit sphere and closed unit balls, respectively, viewed as being embedded in $\mathbb{R}^{n\times d_0}$. Let $x_{i\ast}$ denote the $i$-th row of $x\in \mathbb{R}^{n\times d_0}$.
We aim to solve
\begin{equation}
\max_{x\in(\mathbb{S}^{d_0-1})^n} \Tr(x^\top Q x),\label{Optim}
\end{equation}
where we identify $x\in(\mathbb{S}^{d_0-1})^n$ with $x\in\mathbb{R}^{n\times d_0}$ such that $\| x_{i\ast}\|_2 =1$  for all $i\in\{1,\dots, n\}$.
Since each row of $x$ is of unit $2$-norm, changing the diagonal elements of $Q$ will yield a problem with the same optimal solutions. Namely, the optimal objective will only be shifted by a constant. In view of this remark,  we define the following objective function 
$$f(x) = \Tr(x^\top K x),$$
where 
\begin{equation}
K_{ij}=
\begin{cases}
 Q_{ij} & \text{ if } i\neq j\\
 1+\epsilon + \sum_{k\neq i}|Q_{ik}| & \text{ if } i=j
 \end{cases},\label{eq:Kpreprocess}
\end{equation}
for some $\epsilon>0$.
The result is that $K$ is strictly diagonally dominant with a positive diagonal and therefore $K$ is strictly positive definite. Hence, we have
$
  f(x) >0,
$
for all $x\in \mathbb{R}^{n\times d_0}$.
As a consequence, $\Tr(x^\top K y)$ defines an inner product in $\mathbb{R}^{n\times d_0}$. 
Furthermore, multiplying a vector on the product of unit spheres $(\mathbb{S}^{d_0-1})^n$ by $K$ yields a vector lying in the complement of the closed unit balls $(\mathbb{B}^{d_0 \complement})^n$, as it is stated in Lemma \ref{Lemma:DiagDomin}.
\begin{mylemma}[Effect of diagonal dominance]\label{Lemma:DiagDomin}
Let $x\in (\mathbb{S}^{d_0-1})^n$ and let $K\in \mathbb{R}^{n\times n}$ be a symmetric matrix. 
If the diagonal elements satisfy $|K_{ii}| > 1+\sum_{k\neq i}|K_{ik}|$ for all $1\leq i\leq n$, then we have
$\|(Kx)_{i\ast} \|_2 > 1$ for all $1\leq i\leq n$.
\end{mylemma}
\begin{proof}
 Let $1\leq i\leq n$ and $x_i = x_{i\ast} \in \mathbb{S}^{d_0-1}$. By using successively the reverse triangle inequality $\big|\|a\|-\|b\|\big|\leq \|a-b\|$ and the triangle inequality, we have
\begin{align*}
\|(Kx)_{i\ast} \|_2 &=\big\|K_{ii}x_{i\ast}-(-1)  \sum_{k\neq i} K_{ik} x_{k\ast}\big\|_2\geq \big| |K_{ii}| -\|\sum_{k\neq i} K_{ik}x_{k\ast}\|_2\big|,\\
 &\geq |K_{ii}| -\|\sum_{k\neq i} K_{ik}x_{k\ast}\|_2\geq |K_{ii}| -\sum_{k\neq i} |K_{ik}| > 1.
\end{align*}
\end{proof}
In other words, if $K$ is sufficiently diagonally dominant, all the rows of $Kx$ are vectors with a $2$-norm larger than one. This observation is useful since the iteration of Algorithm~\ref{Alg1} consists of successive multiplications by $K$ and projections on the product of the unit spheres.
By relying on Lemma~\eqref{Lemma:DiagDomin} and on the convexity of a Cartesian product of balls, we can provide a lower bound on the improvement of the objective values between successive iterations; see Proposition~\ref{PropositionMonotone} hereafter.

\section{Generalized power method}\label{appendix:GPM}
The strategy to maximize $f(x)$ is to iteratively maximize linear lower bounds obtained simply as follows
\[
f(x) \geq f(\bar{x}) +2\Tr\left((x-\bar{x})^\top K \bar{x}\right), \text{ for all } \bar{x}\in\mathbb{R}^{n\times d_0},
\]
and which is a consequence of the convexity of $f$ given  that we assumed $K\succ 0$.
More explicitly, by starting from an initial $\bar{x} = x_0 \in(\mathbb{S}^{d_0-1})^n$, the iteration is given by
\begin{equation}
 x_{m+1} = \argmax_{x\in(\mathbb{S}^{d_0-1})^n}\Tr(x^\top K x_m).\label{eq:PPMx}
\end{equation}
Since the maxima of a convex function over a convex set are at the extreme points of this set, and the function $f(x)$ considered here is convex, we can equivalently maximize $f(x)$ on the product of closed unit balls $(\mathbb{B}^{d_0})^n$, that is,
\[
\max_{x\in(\mathbb{S}^{d_0-1})^n} \Tr(x^\top K x) = \max_{x\in(\mathbb{B}^{d_0})^n} \Tr(x^\top K x).
\]
 Hence, the iteration~(\ref{eq:PPMx}) is of the type described in the paper~\cite{Journee:2010:GPM:1756006.1756021} and it is a slight generalization of the algorithm proposed in~\cite{Boumal}.
The iteration~(\ref{eq:PPMx}) is explicitly given by the following projection
\begin{equation}
(x_{m+1})_{i\ast} = \left(\Pi(Kx_{m})\right)_{i\ast} = \frac{(Kx_m)_{i\ast}}{\|(Kx_m)_{i\ast}\|_2},\ \text{ with } 1\leq i\leq n \text{ and }  1\leq \ell\leq d_0,\label{eq:iterate}
\end{equation}
where $\Pi$ is a projection on the product of closed unit balls $(\mathbb{B}^{d_0})^n$. Note that, thanks to the definition of~\eqref{eq:Kpreprocess} and by Lemma~\ref{Lemma:DiagDomin}, we know that for all $x\in(\mathbb{S}^{d_0-1})^n$ we have $\|(Kx)_{i\ast} \|_2 > 1$ for all  $i\in\{1,\dots,n\}$. Thus, we have $(Kx)_{i\ast}\neq 0$ for $i\in\{1,\dots,n\}$.

\subsection{Fixed points are critical points}
Following~\cite{Journee:2010:GPM:1756006.1756021}, we introduce the following first order criterion:
\[
\Delta(\bar{x}) = \max_{x\in (\mathbb{B}^{d_0})^n}\left\langle x-\bar{x}, \nabla f(\bar{x})\right\rangle,
\]
which satisfies clearly $\Delta(\bar{x}) \geq 0$ for all $\bar{x}\in (\mathbb{B}^{d_0})^n$. A critical point satisfies $\Delta(x)  = 0$. 
In particular, we have 
\begin{equation}
\Delta(x) = \sum_{i=1}^n \|(Kx)_{i\ast}\|_2 - \Tr(x^\top K x).\label{eq:Delta}
\end{equation}
Then, we can show that the fixed points of the algorithm are first order critical points of $f(x)$.
\begin{mylemma}[Criterion for criticality]\label{LemEquiv}
Let $x\in (\mathbb{S}^{d_0-1})^n$. The following statements are equivalent: (i) $x = \Pi(Kx)$, and 
(ii) $\Delta(x)  =0$.
\end{mylemma}
\begin{proof}
($(i) \Rightarrow (ii)$) We assume $(x)_{i\ast} \|(Kx)_{i\ast}\|_2= (Kx)_{i\ast}$ and take the inner product with $x_{i}$ which yields
$\sum_{i=1}^n\|(Kx)_{i\ast}\|_2 \sum_{\ell =1}^{d_0} (x)^2_{i\ell}  = \Tr(x^\top K x) $.\\
($(i) \Leftarrow (ii)$) We assume that
$\sum_{i = 1}^n \sum_{\ell =1}^{d_0}  x_{i\ell}(K  x)_{i\ell} = \sum_{i=1}^n \|(Kx)_{i\ast}\|_2$
where by using the  Cauchy-Schwarz inequality, each term of the sum is dominated as follows: \[|\sum_{\ell =1}^{d_0}  x_{i\ell}(K  x)_{i\ell}|\leq \|(Kx)_{i\ast}\|_2.\] Then, we have a vanishing sum of positive terms
\[
\sum_{i = 1}^n \underbrace{\Big( \|(Kx)_{i\ast}\|_2-\sum_{\ell =1}^{d_0}  x_{i\ell}(K  x)_{i\ell}\Big)}_{\geq 0}=0,
\]
and hence each term vanishes.
\end{proof}

\subsection{Monotonicity}
In the paper~\cite{Boumal} which deals with a similar iteration in a different context, it is shown that $f(x_{m+1})\geq f(x_m)$ for all $m\geq 0$. 
We repeat the argument here for completeness.
Indeed, since $x_{m+1}$ is optimal, we have $ \Tr(x^\top K x_m)\leq \Tr(x_{m+1}^\top K x_m)$ for all $x\in(\mathbb{S}^{d_0-1})^n$.
In particular, as a consequence of the Cauchy-Schwarz inequality, we have
\begin{equation}
0\leq \Tr(x_{m+1}^\top K x_m)-f(x_m)\leq \sqrt{f(x_m)f(x_{m+1})}-f(x_m)
\end{equation}
yielding $\sqrt{f(x_m)}\leq \sqrt{f(x_{m+1})}$. In the following proposition, we improve slightly the guarantees of monotonicity of the objective values when $K$ is sufficiently diagonally dominant in the sense of Lemma~\ref{Lemma:DiagDomin}.
\begin{myprop}[Monotonicity of the objectives]\label{PropositionMonotone}
Let $K\in \mathbb{R}^{n\times n}$  symmetric such that $|K_{ii}| > 1+\sum_{k\neq i}|K_{ik}|$ for all $1\leq i\leq n$. Then, the sequence of objectives for the iteration~(\ref{eq:PPMx})  satisfies
\[
\|x_{m+1}-x_m\|_2^{2}\leq  2 \Delta(x_m) < f(x_{m+1})-f(x_m)
\]
for all integers $m\geq 0$.
\end{myprop}
Also, since the objectives are monotone increasing and upper bounded then the objective values $f(x_m)$ converge. An easy consequence is that the stepsize converges as stated in Corollary~\ref{corol:step_size_convergence}. This is simply shown by summing the inequalities in Proposition~\ref{PropositionMonotone} and by taking a limit.
\begin{mycorol}\label{corol:step_size_convergence}
Under the conditions of Proposition~\ref{PropositionMonotone}, we have
\[
\sum_{m=0}^{\infty}\|x_{m+1}-x_m\|_2^{2}\leq  f_{\star}-f(x_{0}),
\]
where $f_{\star} = \lim_{m\to \infty}f(x_{m})$.
\end{mycorol}
The rate of convergence, given in terms of the first order criterion~\eqref{eq:Delta} can also be obtained by summing the inequalities in Proposition~\ref{PropositionMonotone}.
\begin{mycorol}
Let $\Delta_k = \min_{m=1,\dots,k} \Delta(x_m)$.
Then, under the conditions of Proposition~\ref{PropositionMonotone}, we have
\[
\Delta_k\leq \frac{f_{\star}-f(x_{0})}{2k},
\]
where $f_{\star} = \lim_{m\to \infty}f(x_{m})$.
\end{mycorol}
Before proving Proposition~\ref{PropositionMonotone}, we recall a technical result concerning the projection on convex set.
\begin{mylemma}[see e.g.~\cite{Bubeck}]\label{Lemma:ProjConvex}
Let $C$ be a convex set of $\mathbb{R}^n$ and let $\Pi:\mathbb{R}^n\to C$ be the orthogonal projection on $C$. Let $z\in\mathbb{R}^n$ and $x\in C$. Then, we have
\[
\|\Pi(z)-x\|_2^2+\|\Pi(z)-z\|_2^2\leq \|z-x\|_2^2.
\]
\end{mylemma}
We now prove Proposition~\ref{PropositionMonotone}.
\begin{proof}[Proof of Proposition~\ref{PropositionMonotone}]
First, notice that the Cartesian product of convex sets is convex and recall Lemma \ref{Lemma:DiagDomin}. 
Thus, we can use Lemma~\ref{Lemma:ProjConvex} with $C = (\mathbb{B}^{d_0})^n$, and we find
\[
\|\Pi(Kx_m)-x_m\|_2^2\leq \|Kx_m-x_m\|_2^2-\|\Pi(Kx_m)-Kx_m\|_2^2.
\]
The right-hand side above is simplified by expanding the squares, and we obtain
\[
\|x_{m+1}-x_m\|_2^2\leq 2\Tr\left((x_{m+1}-x_m)^\top Kx_m\right) = 2\Delta(x_m),
\]
where the last equality is obtained thanks to~\eqref{eq:Delta} in which the definition of $x_{m+1}$ in \eqref{eq:iterate} is substituted.
Finally, we find 
\begin{align*}
  2\Delta(x_m) = 2\Tr\left((x_{m+1}-x_m)^\top Kx_m\right)&= f(x_{m+1})-f(x_m) -f(x_{m+1}-x_{m})\\
  &< f(x_{m+1})-f(x_m),
\end{align*}
where we used that $f(x)>0$ for all $x\in \mathbb{R}^{n\times d_0}$.
\end{proof}

\subsection{Initialization}
In this paper, we use an initialization procedure inspired by~\cite{ChenCandes}.  The initial $x_0\in (\mathbb{S}^{d_0-1})^n$ is the projection by $\Pi$ of a matrix obtained by sampling uniformly $d_0$ rows of the matrix $Q$. Next, we construct the positive semi-definite matrix $K$ as it is explained in~(\ref{eq:Kpreprocess}). 
\subsection{Generalized Power Method with Momentum}
Finally, we present here a modified version of the iteration analysed in the previous section. This improved algorithm is inspired by the accelerated gradient descent techniques. Empirically, we observe a significant improvement of this iteration with momentum compared to the Generalized Power Method  (GPM) .
Let the sequence $r_m = (m-1)/(m+2)$ for $m\in\{1,2,\dots\}$. 
Let $m\geq 1$ and $x_1\in(\mathbb{S}^{d_0-1})^n$ given. The Generalized Power Method with Momentum  (GPMM)  is given by the following iteration:
\begin{equation}
\begin{cases}
 y_{m} = x_{m}+r_m (x_m - x_{m-1})\\
 x_{m+1} = \argmax_{x\in(\mathbb{S}^{d_0-1})^n}\Tr(x^\top K y_m)
 \end{cases}.\label{Alg}
\end{equation}
As an illustration, we provide in Figure~\ref{Fig:AlgoComparison} a comparison of the convergence of the GPM and GPMM on the ellipsoidal embedding of \textsc{PowerEU} where we observe that GPMM converges faster.
\begin{figure}[h]
\centering\includegraphics[scale = 0.55]{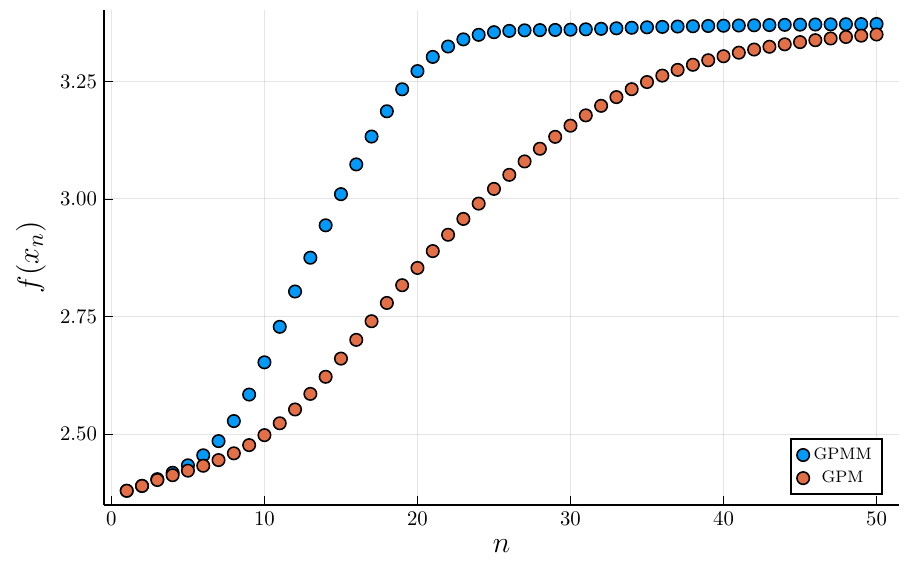}
\caption{Generalized Power Method (GPM) compared to Generalized Power Method with Momentum (GPMM) applied to the ellipsoidal embedding of \textsc{PowerEU} illustrated in Figure~\ref{Fig:Power}.\label{Fig:AlgoComparison}}

\end{figure}
\section{Numerical results}\label{sec:NumericalResults}

\subsection{Hardware}
All the simulations in this paper were performed on a laptop with  $1.1$GHz Dual-Core Intel Core i3 processor and $8$GB RAM.

\subsection{Partitionning of LFR benchmark networks}

To complement the comparison of our partitioning method (ellipsoidal embedding followed by \cref{Alg:VecPart}) and Louvain for the clustering of the LFR benchmark neworks of \cref{Fig:nmi_vs_mu}, we provide in \cref{Fig:nmi_vs_mu_ablation} a study of the influence of the parameters $d_0$ and $k$ on the quality of the retrieved community structure.
We observe that if these two parameter are large enough the NMI between the retrieved and planted partitions is rather stable with respect to variations of these parameters.

\subsection{Partitionning of real networks}\label{sec:Simu}

In this section, we compare the partitions obtained from ellipsoidal embeb-and-partition methods with other community detection methods on the real networks of \cref{Table:RealNetworks}.
To do so, the vector partitioning algorithm is executed on an ellipsoidal embedding computed using the modularity matrix as our descriptor matrix, with a pre-selected upper bound $d_0$ on the embedding dimension.
The partitioning algorithm is runned $5$ times on one embedding and the partition maximizing the objective \cref{eq:ObjectiveVectorPartition} is kept. 
The modularity of this partition is then reported in \cref{fig:ResultsReal}.

We compare our results with the Louvain method~\cite{Louvain} as discussed above, the embedding method node2vec~\cite{node2vec}, as well as a spectral embedding~\cite{ZhangNewman} also based on the modularity matrix.
To derive a clustering from the node2vec embedding, we perform a $32$-dimensional embedding on which  $k$-means clustering is applied.
The number of clusters $k$  ($1\leq k \leq 50$) is selected in order to maximize the modularity of the obtained partition.
In order to avoid storing the full modularity matrix in memory to compute the spectral embedding, we implement a Krylov subspace method by taking advantage of the structure of the modularity matrix (sparse + rank $1$). 
Then, to find clusters, we use again the Vector Partitioning of Algorithm~\ref{Alg:VecPart}.
The dimension of the embedding $n_{\rm ev}$  ($1 \leq n_{\rm ev} \leq 50$) is the one maximizing the modularity of the partition.

The ellipsoidal  and spectral embedding are implemented in Julia.
Node2vec uses the original python implementation of \cite{node2vec}.
For the Louvain method, we use the implementation based on the igraph sofware package~\cite{igraph}. 
Notice that in contrast to the other algorithms, the Louvain method does not provide an embedding of the graph, but only a partitioning.
\begin{figure}[tb!]
  \centering
  \includegraphics[scale=0.5]{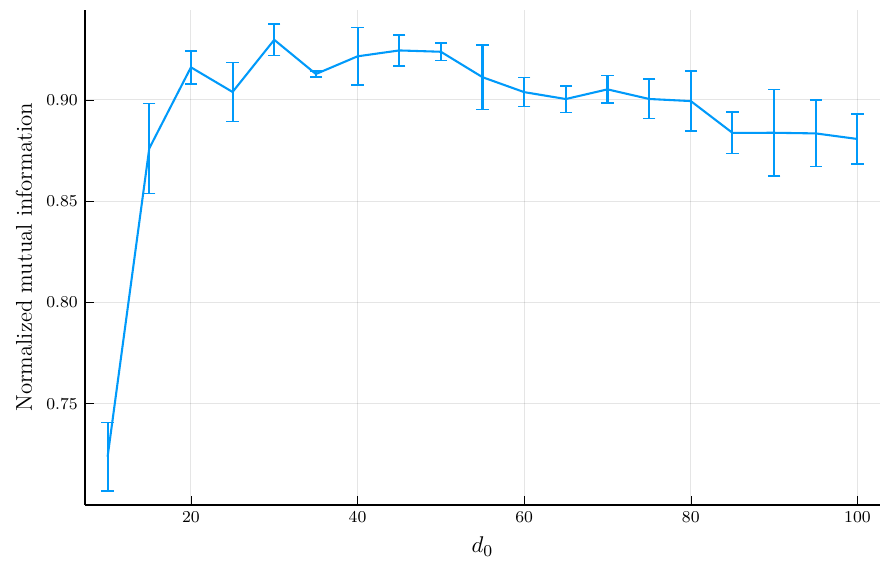}
  \hfill
  \includegraphics[scale=0.5]{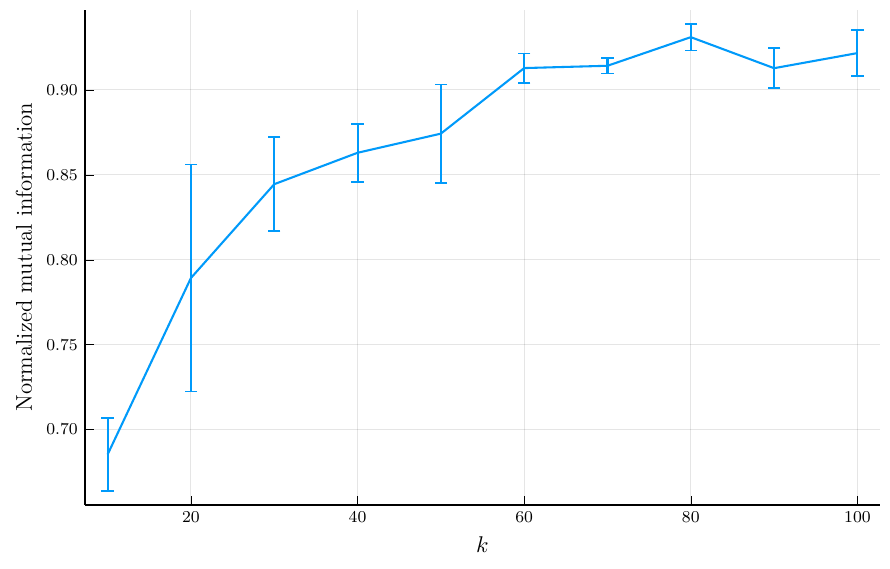}
  \caption{Ablation study for the influence of $d_0$ and $k$ on the results of \cref{Fig:nmi_vs_mu}. NMI \emph{vs} $d_0$ (left, mixing parameter $\texttt{mu}=0.4$ and $k=100$) and NMI \emph{vs} number of initialized clusters $k$ (right, mixing parameter $\texttt{mu}=0.4$ and $d_0 = 30$).
  As in \cref{Fig:nmi_vs_mu}, we display mean and standard deviations over $3$ runs. 
  We refer to \cref{a:nmi_vs_mixing} for parameters associated to the generated LFR benchmark networks. \label{Fig:nmi_vs_mu_ablation}}
\end{figure}
\begin{table}[tb!]
\centering\footnotesize
\begin{tabular}{lccl}
\toprule
  Network & $|\mathcal{V}|$ & $|\mathcal{E}|$ & Description\\
   \midrule
  \textsc{PolBlogs} \cite{Adamic}  & 1221 & 18958 & Blogs about US politics.  \\
  \textsc{PowerEU} \cite{PowerEUrope} & 2712 & 3580 & Power network. \\
  \textsc{Facebook} \cite{McAuley}  & 4039 & 176468 & Social network\\
  \textsc{PowerEU}  \cite{Watts:1998:nature}  & 4941 & 6594 & Power network.\\
  \textsc{Wikivote}  \cite{Leskovec:2010} & 7066 &103663 & Votes between Wikipedia users. \\
  \textsc{Internet}  \cite{WebNewman} & 22963 & 96872& Snapshot of the  Internet.\\
  \textsc{CondMat03}  \cite{Newman404}& 27519 & 411251 & Collaborations.\\
  \textsc{Pokec} \cite{Pokec} &1632803 &30622564 & Social network.\\
  \bottomrule
\end{tabular}
\caption{Features of several real networks used in the study of \cref{fig:ResultsReal}.\label{Table:RealNetworks}}
\end{table}
\begin{table}[tb!]\centering
\resizebox{\columnwidth}{!}{%
\begin{tabular}{lccccccccccccccccccc}
\toprule
   Network & \phantom{abc}& \multicolumn{5}{c}{Ellipsoidal  Embedding $+$ Vector partition }    & \phantom{abc}&  \multicolumn{3}{c}{Louvain} & \phantom{abc}&  \multicolumn{3}{c}{node2vec $+$ $k$-means} & \phantom{abc}&  \multicolumn{3}{c}{Spectral} \\
   \cmidrule{3-7}    \cmidrule{9-11}    \cmidrule{13-15}    \cmidrule{17-20}
  & & $d_{0}$ & ${d}_{\rm  eff}$ & $Q$ & $n_c$  & time (s)  && $Q$ & $n_c$  & time (s) && $Q$ & $n_c$  & time (s) && $n_{\rm ev}$ & $Q$ & $n_c$ & time (s)\\
   \midrule
 \textsc{PolBlogs} & & 10 & 2 &  {\bf 0.43} & 9  & 0.2     && {\bf 0.43} & 11  & 0.1 && {\bf 0.43} & 2  & $25$ && 2 & 0.40  & 3 & 9\\
   \midrule
  \textsc{PowerEU} & & 50  & 3  & 0.91 & 37 & 4.1   && {\bf 0.92} & 30 &  0.1 && 0.91 &  38  & $20$ && 32 & 0.84  & 38 & 37\\
   \midrule
  \textsc{Facebook} & & 50 & 3  & 0.82 & 32 & 10.1    && {\bf 0.83} & 17  &  0.4 && 0.82 & 9  & $94$ && 8 & 0.79  & 9 & 133\\
   \midrule   
      \textsc{PowerUS} & & 50 & 3  & 0.92 & 34 & 10.1     &&  {\bf 0.94} & 44  &  0.1 && 0.93 & 35  & $36$ && 48 & 0.72  & 38 & 243\\
   \midrule
    \textsc{Wikivote} &  & 30 & 2  & {\bf 0.42} & 9 & 1.1    &&  {\bf 0.42} & 7  & 0.22 && 0.41  & 5  & $212$ && 3 & {\bf 0.42}  & 5 & 342 \\
   \midrule
    \textsc{Internet} & & 70 & 3   &  0.62 & 15 & 177   && {\bf 0.66} & 37 & 0.6 && 0.64 & 17  & $469$ && 45 & 0.50 & 32 & 782\\
   \midrule
    \textsc{CondMat03}$^\star$ & & 250 & 6  & 0.69 & 67 & 658   && {\bf 0.74} & 64 & 1.2 && 0.70 & 47 & $397$ && 50 & 0.51 & 40 & 936\\
    \midrule
   \textsc{Pokec} & & 50 &4   & 0.72 & 19  & 6343  &&  {\bf 0.73} & 47 &  658 && {\slash} & {\slash}  & {\slash} && 50 & 0.61 & 35 & 40384\\
  \bottomrule
\end{tabular}
}
\caption{Clustering of the networks of \cref{Table:RealNetworks}. Here $n_c$ is the number of communities obtained while $Q$ is the modularity of the associated partition. The running time corresponds to one execution.
Clusterings of the ellipsoidal embeddings are obtained by initializing $k=100$ centroids in the vector partitioning algorithm; see \cref{Alg:VecPart} (The $^\star$ indicates that $k=200$ centroids were used). 
The final partition is the one maximizing the objective \cref{eq:ObjectiveVectorPartition} over $5$ executions of \cref{Alg:VecPart}.
The Louvain method was executed once.
Node2vec was used to generate a $32$-dimensional embedding and the partition obtained by $k$-means with the best modularity is reported.
Notice that we could not use node2vec on the large \textsc{Pokec} network due to memory issues. 
For the spectral algotithm, $n_{\rm ev}$ denotes the number of dominant eigenvectors used. 
We only report the time necessary to compute the spectral embedding.\label{fig:ResultsReal}}
\end{table}

\subsection{Simulation of \cref{Fig:embed_toy}}
\label{a:fig_embed_toy}
We generated the benchmark networks \textsc{LFR1} and \textsc{LFR2} of \cref{Fig:embed_toy} thanks to the function \texttt{LFR\_benchmark\_graph} of the \texttt{networkx} package in Julia, with the parameters given in \cref{t:param_toy_LFR} and
$\texttt{average\_degree}$ $=\texttt{None}$,
$\texttt{min\_degree}=20$,
$\texttt{max\_degree}=50$,
$\texttt{tol}=1\mathrm{e}{-07}$,
$\texttt{max\_iters}=500$,
$\texttt{seed}=0$,
$\texttt{max\_community}=1000$.
\begin{table}[tb!]\centering
\begin{tabular}{lccccc}
\toprule
Graph & \texttt{n} & \texttt{mu} & \texttt{tau1} & \texttt{tau2} & \texttt{min\_community}\\
\midrule
\textsc{LFR1}& 2000 & 0.1 & 2 & 1.1 & 200\\
\textsc{LFR2}& 2000 & 0.2 & 2 & 3 & 100\\
\bottomrule
\end{tabular}
\caption{Parameters used for generating the networks of \cref{Fig:embed_toy}.\label{t:param_toy_LFR}}
\end{table}

\subsection{Simulations of \cref{Fig:nmi_vs_mu}\label{a:nmi_vs_mixing}}
The LFR networks of \cref{Fig:nmi_vs_mu} were generated tahnks to the function \texttt{LFR\_benchmark\_graph} of the \texttt{networkx} package with the following parameters:
$\texttt{n}=1000$,
$\texttt{tau1}=2$,
$\texttt{tau2}=2$,
$\texttt{average\_degree}=15$,
$\texttt{min\_degree}=\texttt{None}$,
$\texttt{max\_degree}=50$,
$\texttt{min\_community}=50$,
$\texttt{max\_community}=\texttt{None}$,
$\texttt{tol}=1\mathrm{e}{-07}$,
$\texttt{max\_iters}=500$,
$\texttt{seed}=0$.

\FloatBarrier
\bibliographystyle{unsrt}
\bibliography{Bib_VectorPartition-4}

\begin{thebibliography}{10}

\bibitem{Strogatz2001}
S.~H. Strogatz.
\newblock {E}xploring complex networks.
\newblock {\em Nature}, 410(6825):268--276, March 2001.

\bibitem{Newman2003}
M.~E.~J. Newman.
\newblock The structure and function of complex networks.
\newblock {\em SIAM review}, 45(2):167--256, 2003.

\bibitem{Boccaletti2006}
S.~Boccaletti, V.~Latora, Y.~Moreno, M.~Chavez, and D.-U. Hwang.
\newblock {C}omplex networks: {S}tructure and dynamics.
\newblock {\em Physics Reports}, 424(4-5):175--308, 2006.

\bibitem{Arenas2008a}
A.~Arenas, A.~D{\'{\i}}az-Guilera, J.~Kurths, Y.~Moreno, and C.~Zhou.
\newblock {S}ynchronization in complex networks.
\newblock {\em Physics Reports}, 469(3):93--153, 2008.

\bibitem{Dorogovtsev2008}
S.~N. Dorogovtsev, A.~V. Goltsev, and J.~F.~F. Mendes.
\newblock {C}ritical phenomena in complex networks.
\newblock {\em Rev. Mod. Phys.}, 80:1275--1335, October 2008.

\bibitem{Sporns2009}
O.~Sporns and E.~Bullmore.
\newblock Complex brain networks: graph theoretical analysis of structural and
  functional systems.
\newblock {\em Nat Rev Neurosci.}, 10, 2009.

\bibitem{AlbertBarabasi}
R.~Albert and A.-L. Barab\'asi.
\newblock Statistical mechanics of complex networks.
\newblock {\em Rev. Mod. Phys.}, 74:47--97, Jan 2002.

\bibitem{QinRohe}
T.~Qin and K.~Rohe.
\newblock Regularized spectral clustering under the degree-corrected stochastic
  blockmodel.
\newblock In {\em Proceedings of the 26th International Conference on Neural
  Information Processing Systems - Volume 2}, NIPS'13, pages 3120--3128, 2013.

\bibitem{Luxburg2007}
U.~von Luxburg.
\newblock {A} tutorial on spectral clustering.
\newblock {\em Statistics and Computing}, 17(4):395--416, 2007.

\bibitem{Rohe2011}
K.~Rohe, S.~Chatterjee, B.~Yu, et~al.
\newblock Spectral clustering and the high-dimensional stochastic blockmodel.
\newblock {\em The Annals of Statistics}, 39(4):1878--1915, 2011.

\bibitem{RDPgraphs}
A.~Athreya, D.~E. Fishkind, M.~Tang, C.~E. Priebe, Y.~Park, J.~T. Vogelstein,
  K.~Levin, V.~Lyzinski, Y.~Qin, and D.~L. Sussman.
\newblock Statistical inference on random dot product graphs: a survey.
\newblock {\em Journal of Machine Learning Research}, 18(226):1--92, 2018.

\bibitem{Louvain}
V.~D. Blondel, J.-L. Guillaume, R.~Lambiotte, and E.~Lefebvre.
\newblock Fast unfolding of communities in large networks.
\newblock {\em Journal of Statistical Mechanics: Theory and Experiment},
  2008(10):P10008, 2008.

\bibitem{ShiMalik}
J.~Shi and J.~Malik.
\newblock Normalized cuts and image segmentation.
\newblock {\em IEEE Trans. Pattern Anal. Mach. Intell.}, 22(8):888--905, August
  2000.

\bibitem{Newman2013}
M.~E.~J. Newman.
\newblock Spectral methods for community detection and graph partitioning.
\newblock {\em Phys. Rev. E}, 88:042822, 2013.

\bibitem{Luxburg2008}
U.~von Luxburg, M.~Belkin, and O.~Bousquet.
\newblock Consistency of spectral clustering.
\newblock {\em The Annals of Statistics}, 36(2):555--586, apr 2008.

\bibitem{Saerens2004}
M.~Saerens, F.~Fouss, L.~Yen, and P.~Dupont.
\newblock {T}he {P}rincipal {C}omponents {A}nalysis of a {G}raph, and {I}ts
  {R}elationships to {S}pectral {C}lustering.
\newblock In {\em Machine Learning: ECML 2004}, volume 3201 of {\em Lecture
  Notes in Computer Science}, pages 371--383. Springer Berlin / Heidelberg,
  2004.

\bibitem{zhang2008}
Z.~Zhang, M.~I. Jordan, et~al.
\newblock Multiway spectral clustering: A margin-based perspective.
\newblock {\em Statistical Science}, 23(3):383--403, 2008.

\bibitem{Coifman2005}
R.~R. Coifman, S.~Lafon, A.~B. Lee, M.~Maggioni, B.~Nadler, F.~Warner, and
  S.~W. Zucker.
\newblock Geometric diffusions as a tool for harmonic analysis and structure
  definition of data: Diffusion maps.
\newblock {\em Proceedings of the National Academy of Sciences},
  102(21):7426--7431, 2005.

\bibitem{Lafon2006}
S.~Lafon and A.B. Lee.
\newblock {D}iffusion maps and coarse-graining: a unified framework for
  dimensionality reduction, graph partitioning, and data set parameterization.
\newblock {\em Pattern Analysis and Machine Intelligence, IEEE Transactions
  on}, 28(9):1393--1403, September 2006.

\bibitem{Nadler2006}
B.~Nadler, S.~Lafon, R.R. Coifman, and I.~G. Kevrekidis.
\newblock {D}iffusion maps, spectral clustering and reaction coordinates of
  dynamical systems.
\newblock {\em Applied and Computational Harmonic Analysis}, 21(1):113--127,
  2006.
\newblock Diffusion Maps and Wavelets.

\bibitem{Asta2015}
D.~M. Asta and C.~R. Shalizi.
\newblock Geometric network comparisons.
\newblock In {\em Proceedings of the Thirty-First Conference on Uncertainty in
  Artificial Intelligence}, UAI'15, pages 102--110, Arlington, Virginia, United
  States, 2015. AUAI Press.

\bibitem{Lovasz}
L.~Lov{\'{a}}sz.
\newblock {\em Large Networks and Graph Limits}, volume~60 of {\em Colloquium
  Publications}.
\newblock American Mathematical Society, 2012.

\bibitem{Abbe}
E.~Abbe and C.~Sandon.
\newblock Community detection in general stochastic block models: Fundamental
  limits and efficient algorithms for recovery.
\newblock In {\em 2015 IEEE 56th Annual Symposium on Foundations of Computer
  Science}, pages 670--688, Oct 2015.

\bibitem{PHR-D22}
F.~Sanna~Passino, N.~A. Heard, and P.~Rubin-Delanchy.
\newblock Spectral clustering on spherical coordinates under the
  degree-corrected stochastic blockmodel.
\newblock {\em Technometrics}, 0(0):1--12, 2022.

\bibitem{GvdHL21}
Martijn G{\~A}{\c{k}}sgens, Remco van~der Hofstad, and Nelly Litvak.
\newblock The hyperspherical geometry of community detection: modularity as a
  distance.
\newblock {\em Journal of Machine Learning Research}, 24(112):1--36, 2023.

\bibitem{Gutierrez}
L.~Guti\'errez~G\'omez, B.~Chi\^em, and J.-C. Delvenne.
\newblock Dynamics based features for graphs classification.
\newblock {\em arxiv:1705.10817, submitted}.

\bibitem{Hamilton2017}
W.~L. Hamilton, R.~Ying, and J.~Leskovec.
\newblock Representation learning on graphs: Methods and applications.
\newblock {\em arXiv preprint arXiv:1709.05584}, 2017.

\bibitem{node2vec}
A.~Grover and J.~Leskovec.
\newblock node2vec: Scalable feature learning for networks.
\newblock In {\em ACM SIGKDD International Conference on Knowledge Discovery
  and Data Mining (KDD)}, 2016.

\bibitem{Chung:1997}
F.~R.~K. Chung.
\newblock {\em Spectral Graph Theory}.
\newblock American Mathematical Society, 1997.

\bibitem{Mohar91}
B.~Mohar.
\newblock The laplacian spectrum of graphs.
\newblock In {\em Graph Theory, Combinatorics, and Applications}, pages
  871--898. Wiley, 1991.

\bibitem{Chan1994}
P.~K. Chan, M.~D.~F. Schlag, and J.~Y. Zien.
\newblock Spectral k-way ratio-cut partitioning and clustering.
\newblock {\em IEEE Transactions on Computer-Aided Design of Integrated
  Circuits and Systems}, 13(9):1088--1096, 1994.

\bibitem{Peixoto}
T.~P. Peixoto.
\newblock Nonparametric bayesian inference of the microcanonical stochastic
  block model.
\newblock {\em Phys. Rev. E}, 95:012317, Jan 2017.

\bibitem{Amini2018}
A.~A. Amini and E.~Levina.
\newblock On semidefinite relaxations for the block model.
\newblock {\em Ann. Statist.}, 46(1):149--179, 02 2018.

\bibitem{Hajek2016}
B.~Hajek, Y.~Wu, and J.~Xu.
\newblock Achieving exact cluster recovery threshold via semidefinite
  programming.
\newblock {\em IEEE Transactions on Information Theory}, 62(5):2788--2797, May
  2016.

\bibitem{MontanariPNAS}
A.~Javanmard, A.~Montanari, and F.~Ricci-Tersenghi.
\newblock Phase transitions in semidefinite relaxations.
\newblock {\em PNAS}, 113(16):E2218--E2223, 2016.

\bibitem{Boumal}
N.~Boumal.
\newblock Nonconvex phase synchronization.
\newblock {\em SIAM Journal on Optimization}, 26(No. 4):2355--2377, 2016.

\bibitem{Singer2011}
A.~Singer.
\newblock Angular synchronization by eigenvectors and semidefinite programming.
\newblock {\em Applied and computational harmonic analysis}, 30(1):20, 2011.

\bibitem{Boumal:2016}
N.~Boumal, V.~Voroninski, and A.~S. Bandeira.
\newblock The non-convex burer--monteiro approach works on smooth semidefinite
  programs.
\newblock In {\em Proceedings of the 30th International Conference on Neural
  Information Processing Systems}, NIPS'16, pages 2765--2773, 2016.

\bibitem{Brandes2006}
U.~Brandes, D.~Delling, M.~Gaertler, R.~Goerke, M.~Hoefer, Z.~Nikoloski, and
  D.~Wagner.
\newblock Maximizing modularity is hard, 2006.

\bibitem{NewmanSpectral}
M.~E.~J. Newman.
\newblock Finding community structure in networks using the eigenvectors of
  matrices.
\newblock {\em Phys. Rev. E}, 74:036104, Sep 2006.

\bibitem{Journee:2010:GPM:1756006.1756021}
M.~Journ{\'e}e, Y.~Nesterov, P.~Richt\'{a}rik, and R.~Sepulchre.
\newblock Generalized power method for sparse principal component analysis.
\newblock {\em J. Mach. Learn. Res.}, 11:517--553, March 2010.

\bibitem{ChenCandes}
Y.~Chen and E.~Candes.
\newblock The projected power method: An efficient algorithm for joint
  alignment from pairwise differences, arxiv:1609.05820.
\newblock 2016.

\bibitem{MasterThesis}
A.~Aspeel.
\newblock {Community Detection in Large-Scale Time-Varying Networks, A
  Modularity Based Approach; Master thesis, Universit\'e catholique de
  Louvain}, 2017.

\bibitem{Nesterov1983}
Y.~E. Nesterov.
\newblock A method for solving the convex programming problem with convergence
  rate $\mathcal{O}(1/k^2)$.
\newblock In {\em Dokl. Akad. Nauk SSSR}, volume 269, pages 543--547, 1983.

\bibitem{SDPEmbedding}
Micha{\"e}l Fanuel, Antoine Aspeel, Jean-Charles Delvenne, and Johan~AK
  Suykens.
\newblock Positive semi-definite embedding for dimensionality reduction and
  out-of-sample extensions.
\newblock {\em SIAM Journal on Mathematics of Data Science}, 4(1):153--178,
  2022.

\bibitem{Delvenne12755}
J.-C. Delvenne, S.~N. Yaliraki, and M.~Barahona.
\newblock Stability of graph communities across time scales.
\newblock {\em Proceedings of the National Academy of Sciences},
  107(29):12755--12760, 2010.

\bibitem{Schaub2012b}
M.~T. Schaub, J.-C. Delvenne, S.~N. Yaliraki, and M.~Barahona.
\newblock {M}arkov {D}ynamics as a {Z}ooming {L}ens for {M}ultiscale
  {C}ommunity {D}etection: {N}on {C}lique-{L}ike {C}ommunities and the
  {F}ield-of-{V}iew {L}imit.
\newblock {\em PLoS ONE}, 7(2):e32210, 2 2012.

\bibitem{Delvenne2013}
J.-C. Delvenne, M.~T. Schaub, S.~N. Yaliraki, and M.~Barahona.
\newblock {T}he {S}tability of a {G}raph {P}artition: {A} {D}ynamics-{B}ased
  {F}ramework for {C}ommunity {D}etection.
\newblock In Animesh Mukherjee, Monojit Choudhury, Fernando Peruani, Niloy
  Ganguly, and Bivas Mitra, editors, {\em Dynamics On and Of Complex Networks,
  Volume 2}, Modeling and Simulation in Science, Engineering and Technology,
  pages 221--242. Springer New York, 5 2013.

\bibitem{Schaub2018a}
M.~T. Schaub, J.-C. Delvenne, R.~Lambiotte, and M.~Barahona.
\newblock Multiscale dynamical embeddings of complex networks.
\newblock submitted, April 2018.

\bibitem{ZhangNewman}
X.~Zhang and M.~E.~J. Newman.
\newblock Multiway spectral community detection in networks.
\newblock {\em Phys. Rev. E}, 92:052808, Nov 2015.

\bibitem{LiuBarahona}
Z.~Liu and M.~Barahona.
\newblock Geometric multiscale community detection: Markov stability and vector
  partitioning.
\newblock {\em Journal of Complex Networks}, page cnx028, 2017.

\bibitem{LFBenhmarks}
A.~Lancichinetti, S.~Fortunato, and F.~Radicchi.
\newblock Benchmark graphs for testing community detection algorithms.
\newblock {\em Phys. Rev. E}, 78:046110, Oct 2008.

\bibitem{Serrano2008}
M.~A. Serrano, D.~Krioukov, and M.~Bogun{\'a}.
\newblock Self-similarity of complex networks and hidden metric spaces.
\newblock {\em Physical review letters}, 100(7):078701, 2008.

\bibitem{kempe2003quantum}
Julia Kempe.
\newblock Quantum random walks: an introductory overview.
\newblock {\em Contemporary Physics}, 44(4):307--327, 2003.

\bibitem{Bubeck}
S.~Bubeck.
\newblock Convex optimization: Algorithms and complexity.
\newblock {\em Found. Trends Mach. Learn.}, 8(3-4):231--357, November 2015.

\bibitem{igraph}
G.~Csardi and T.~Nepusz.
\newblock The igraph software package for complex network research.
\newblock {\em InterJournal}, Complex Systems:1695, 2006.

\bibitem{Adamic}
L.~A. Adamic and N.~Glance.
\newblock The political blogosphere and the 2004 u.s. election: Divided they
  blog.
\newblock In {\em Proceedings of the 3rd International Workshop on Link
  Discovery}, LinkKDD '05, pages 36--43, New York, NY, USA, 2005. ACM.

\bibitem{PowerEUrope}
M.~Rosas-Casals, S.~Valverde, and R.~V. Sol{\'e}.
\newblock Topological vulnerability of the european power grid under errors and
  attacks.
\newblock {\em International Journal of Bifurcation and Chaos},
  17(07):2465--2475, 2007.

\bibitem{McAuley}
J.~McAuley and J.~Leskovec.
\newblock Learning to discover social circles in ego networks.
\newblock In {\em Proceedings of the 25th International Conference on Neural
  Information Processing Systems - Volume 1}, NIPS'12, pages 539--547, USA,
  2012. Curran Associates Inc.

\bibitem{Watts:1998:nature}
D.~J. Watts and S.~H. Strogatz.
\newblock Collective dynamics of 'small-world' networks.
\newblock {\em Nature}, 393(6684):440--442, June 1998.

\bibitem{Leskovec:2010}
J.~Leskovec, D.~Huttenlocher, and J.~Kleinberg.
\newblock Signed networks in social media.
\newblock In {\em Proceedings of the SIGCHI Conference on Human Factors in
  Computing Systems}, CHI '10, pages 1361--1370, New York, NY, USA, 2010. ACM.

\bibitem{WebNewman}
http://www-personal.umich.edu/\%7emejn/netdata/.

\bibitem{Newman404}
M.~E.~J. Newman.
\newblock The structure of scientific collaboration networks.
\newblock {\em Proceedings of the National Academy of Sciences},
  98(2):404--409, 2001.

\bibitem{Pokec}
L.~Takac and M.~Zabovsky.
\newblock Data analysis in public social networks.
\newblock In {\em Int. Sci. Conf. and Int. Workshop Present Day Trends of
  Innovations}, pages 1--6, 2012.

\end{thebibliography}

\end{document}